\documentclass[10pt,letter]{article}
\usepackage[papersize={8.5in,11in},margin=1in]{geometry}
\usepackage{amsmath,amssymb,amsthm,fullpage}
\usepackage[lined,boxed,ruled,norelsize]{algorithm2e}
\usepackage[colorinlistoftodos,textwidth=4cm]{todonotes} % marking things to do in a LaTeX document
\usepackage{enumitem}

\usepackage{complexity}
\usepackage{subfig}
\usepackage[format=hang,justification=centering]{caption}%
\usepackage[pagebackref=true,colorlinks]{hyperref}
\hypersetup{linkcolor=red,filecolor=red,citecolor=red,urlcolor=red}
\newtheorem{theorem}{Theorem}[section]
\newtheorem{lemma}[theorem]{Lemma}

\newtheorem{claim}[theorem]{Claim}
\newtheorem{definition}[theorem]{Definition} % TODO: this right???

\newcommand{\namedref}[2]{\hyperref[#2]{#1~\ref*{#2}}}

\newcommand{\sectionref}[1]{\namedref{Section}{#1}}
\newcommand{\appendixref}[1]{\namedref{Appendix}{#1}}
\newcommand{\theoremref}[1]{\namedref{Theorem}{#1}}

\newcommand{\figureref}[1]{\namedref{Figure}{#1}}
\newcommand{\figurerefb}[2]{\hyperref[#1]{Figure~\ref*{#1}#2}}

\newcommand{\lemmaref}[1]{\namedref{Lemma}{#1}}
\newcommand{\claimref}[1]{\namedref{Claim}{#1}}

\newcommand{\algorithmref}[1]{\namedref{Algorithm}{#1}}

\newcommand{\equationref}[1]{\hyperref[#1]{(\ref*{#1})}}

\renewcommand{\R}{\mathbb{R}}

\newcommand{\redgevec}{\R^{E}}
\newcommand{\rvertvec}{\R^{V}}
\newcommand{\rPos}{\R^{+}}
\newcommand{\rNonNeg}{\R^{\geq 0}}

\newcommand{\bvar}[1]{\mathbf{#1}} % bold variable
\newcommand{\mvar}[1]{\bvar{#1}} % matrix variable
\newcommand{\vvar}[1]{\vec{#1}} % vector variable

\newcommand{\defeq}{\stackrel{\mathrm{\scriptscriptstyle def}}{=}}

\DeclareMathOperator*{\argmin}{arg\,min}

\newcommand{\lap}{\mvar{L}}
\newcommand{\pseudo}[1]{{#1}^\dagger}
\newcommand{\lapPseudo}{\pseudo{\lap}}

\newcommand{\incMatrix}{\mvar{B}}

\newcommand{\rMatrix}{\mvar{R}} % resistance matrix
\newcommand{\iMatrix}{\mvar{I}} % identity matrix

\newcommand{\tree}{T}
\newcommand{\treePath}[1]{P_{#1}}
\newcommand{\treePathVec}[1]{\vvar{p}_{#1}}
\newcommand{\treeCycle}[1]{C_{#1}}
\newcommand{\treeCycleVec}[1]{\vvar{c}_{#1}}
\newcommand{\st}{\mathrm{st}}
\newcommand{\stretchEdge}[1]{\st\left(#1\right)}
\newcommand{\stretchTotal}[1]{\st\left(#1\right)}
\newcommand{\offtreeEdgeSet}{E \setminus \tree}

\newcommand{\treeCondition}{\tau}
\newcommand{\cyclePotential}[2]{\Delta_{c_{#1}}(#2)}
\newcommand{\cycleResistance}[1]{R_{#1}}

\newcommand{\sampleProbVec}{\vvar{p}}
\newcommand{\edgeSampleProb}[1]{{{\sampleProbVec}_{#1}}}
\newcommand{\nIter}{K}

\newcommand{\dsOpQuery}{\texttt{query}}
\newcommand{\dsOpUpdate}{\texttt{update}}
\newcommand{\dsOpInit}{\texttt{init}}
\newcommand{\dsOpTreeDecompose}{\texttt{tree-decompose}}

\newcommand{\dsVarHeight}{\texttt{height}}

\newcommand{\dsVarExt}{\texttt{\dsVarSeperator\_ext}}
\newcommand{\dsVarTotal}{\texttt{\dsVarSeperator\_drop}}
\newcommand{\dsVarSource}{s}
\newcommand{\dsVarSeperator}{d}

\newcommand{\dsVectorSpace}{\R^{\{e, d\} \times N}}
\newcommand{\dsVectorQuery}[1]{\vvar{q}(#1)}
\newcommand{\dsVectorUpdate}[1]{\vvar{u}(#1)}
\newcommand{\dsAssign}{:=}

\newcommand{\boundary}{\vvar{\chi}}
\newcommand{\flow}{\vvar{f}}
\newcommand{\volt}{\vvar{v}}

\newcommand{\optVec}[1]{{#1}_{\mathrm{opt}}}
\newcommand{\optFlow}{\optVec{\flow}}
\newcommand{\optVolt}{\optVec{\volt}}

\newcommand{\flowInitial}{\flow_0}

\newcommand{\renergyOp}{\xi_r}
\newcommand{\renergy}[1]{\renergyOp\big(#1\big)}
\newcommand{\renergyFull}[1]{\renergyOp\left(#1\right)}
\newcommand{\renergyDualOp}{\zeta_r}
\newcommand{\renergyDual}[1]{\renergyDualOp\big(#1\big)}
\newcommand{\gap}{\mathrm{gap}}
\newcommand{\potentialdrop}[2]{\Delta_{#1}\left(#2\right)}

\newcommand{\runtimeSimple}{O\big(m \log^2 n \log \log n \log (\epsilon^{-1} n)\big)}
\newcommand{\runtimeBetter}{O\big(m \log^2 n \log \log n \log (\epsilon^{-1} \log n)\big)}
\newcommand{\runtimeFinal}{O\big(m \log^2 n \log \log n \log (\epsilon^{-1})\big)}
\newcommand{\runtimeBest}{\runtimeFinal}

\newcommand{\circVec}{\vvar{c}} % Variable for an arbitrary circulation
\renewcommand{\E}{\mathbb{E}}
\renewcommand{\epsilon}{\varepsilon}
\newcommand{\Tr}{\mathrm{Tr}}

\newcommand{\finalSolver}{\texttt{FullSolver}}
\newcommand{\betterSolver}{\texttt{ExampleSolver}}
\newcommand{\simpleSolver}{\texttt{SimpleSolver}}

\begin{document}

\title{A Simple, Combinatorial Algorithm for Solving SDD Systems in Nearly-Linear Time}

\date{\today}
\author{
\makebox[.2\columnwidth]{Jonathan A. Kelner} \\
\texttt{\href{mailto:kelner@mit.edu}{\color{black}kelner@mit.edu}} \\
MIT
\and
\makebox[.2\columnwidth]{Lorenzo Orecchia}\\
\texttt{\href{mailto:orecchia@mit.edu}{\color{black}orecchia@mit.edu}} \\
MIT
\and
\makebox[.2\columnwidth]{Aaron Sidford} \\
\texttt{\href{mailto:sidford@mit.edu}{\color{black}sidford@mit.edu}} \\
MIT
\and
Zeyuan Allen Zhu \\
\texttt{\href{mailto:zeyuan@csail.mit.edu}{\color{black}zeyuan@csail.mit.edu}} \\
MIT
}
\date{}

\maketitle
\begin{abstract}
In this paper, we present a simple combinatorial algorithm that solves symmetric diagonally dominant (SDD) linear systems in nearly-linear time. It uses very little of the machinery that previously appeared to be necessary for a such an algorithm.  It does not require recursive preconditioning, spectral sparsification, or even the Chebyshev Method or Conjugate Gradient.   After constructing a ``nice'' spanning tree of a graph associated with the linear system, the entire algorithm consists of the repeated application of a simple (non-recursive) update rule, which it implements using a lightweight data structure.  The algorithm is numerically stable and can be implemented without the increased bit-precision required by previous solvers.  As such, the algorithm has the fastest known running time under the standard unit-cost RAM model. We hope that the simplicity of the algorithm and the insights yielded by its analysis will be useful in both theory and practice.
\end{abstract}

\section{Introduction}
A matrix $\mvar{A} \in \R^{n \times n}$ is  \emph{symmetric diagonally
dominant} (SDD) if $\mvar{A}^T = \mvar{A}$ and
$\mvar{A}_{ii} \geq \sum_{j \neq i}|\mvar{A}_{ij}|$
for all $i \in [n]$.
While the best known algorithm for solving a general linear system takes time $O(n^{2.373})$~\cite{DBLP:conf/stoc/Williams12}, a seminal paper by Spieman and Teng~\cite{ST04} showed that when $\mvar{A}$ is SDD one can solve $\mvar{A} \vvar{x} = \vvar{b}$ approximately in nearly linear time.
\footnote{Throughout this paper we are primarily interested in \emph{approximate linear system solvers}, that is algorithms that can compute $\vvar{x} \in \R^{n}$ such that $\|\vvar{x} - \optVec{\vvar{x}}\|_{A} \leq \epsilon \|\optVec{\vvar{x}}\|_{A}$ for any $\epsilon \in \R > 0$ where $\optVec{x} \in \R^{n}$ is a vector such that $\mvar{A} \optVec{\vvar{x}} = \vvar{b}$. When we refer to a \emph{nearly linear time SDD system solver} we mean an an algorithm that computes such a $x$ in time $O(m \log^{c} n \log \epsilon^{-1})$ where $m$ is the number of nonzero entries in $A$ and $c \geq 0 \in \R$ is a fixed constant.}

Fast algorithms for solving SDD linear systems have found broad applications across both the theory and practice of computer science. They have long been central to scientific computing, where solving SDD systems is the main computational task in modeling of electrical networks of resistors and performing finite element simulations of a wide range of physical systems~(see, e.g., \cite{boman2008solving}). Beyond this, SDD system solvers have been applied to foundational problems in a wide range of other fields, including machine learning, random processes, computer vision, image processing, network analysis, and computational biology (see, for example, \cite{LLDM,DBLP:journals/cviu/KoutisMT11,isorank,VoevTeng,HLK}).

More recently, SDD solvers have emerged as a powerful tool in the design of graph algorithms.
To every graph $G$, one can associate an SDD matrix $\lap = \lap_G$ called its \emph{Laplacian} (defined in \sectionref{sec:preliminaries}) such that there are deep connections between the combinatorial properties of $G$ and the linear algebraic properties of $\lap$.
By exploiting these connections, researchers have used nearly linear time algorithms for solving SDD systems to break longstanding barriers and provide new algorithms for a wide and rapidly growing list of fundamental graph problems, including maximum flow problems~\cite{CKMST}, multi-commodity flow problems~\cite{KelnerMillerPeng}, generating random spanning tree~\cite{KelnerMadry}, graph sparsification~\cite{spielman2011graph}, lossy flow problems~\cite{DaitchSpielman}, sparsest cut~\cite{ShermanBreaking}, distributed routing~\cite{KelnerMaymounkov}, and balanced separator~\cite{OSV}, as well as fundamental linear algebraic problems for SDD matrices, including computing the matrix exponential~\cite{OSV} and the largest eigenvalue and corresponding eigenvector~\cite{ST08c}. For surveys of these solvers and their applications, see~\cite{Spielman:2012:AGT:2359888.2359901,DBLP:conf/tamc/Teng10, Lxequalsb}.

\subsection{Previous Nearly Linear Time Algorithms}
The first nearly linear time algorithm for solving SDD systems was given by Spielman and Teng~\cite{ST04}, building on a long line of previous work (e.g., \cite{Vaidya,DBLP:conf/ipps/GrembanMZ95,SupportGraph,SupportTheory,boman2004maximum}).
Their algorithm and its analysis is a technical tour-de-force that required multiple fundamental innovations in spectral and combinatorial graph theory, graph algorithms, and computational linear algebra. Their work included the invention of spectral sparsification and ultra-sparsifiers, better and faster constructions of low-stretch spanning trees, and efficient local clustering algorithms, all of which was used to construct and analyze an intricate recursively preconditioned iterative solver. They divided this work into three papers totaling over 130 pages (\cite{ST08a,ST08c,SpielmanTeng_sparsification}), each of which has prompted a new line of inquiry and substantial follow-up work. Their work was was followed by two beautifully insightful papers by Koutis, Miller, and Peng that simplified the SDD system solver while improving its running time
to $O(m \log n \log\log n  \log \epsilon^{-1})$~\cite{KoutisMP10,KMP11}. For a more in-depth discussion of the history of this work see \cite{ST08c}.

These algorithms all rely on the same general framework.  They reduce solving general SDD systems to solving systems in graph Laplacians.  Given a graph, they show how to obtain a sequence of logarithmically many successively sparser graphs that approximate it, which they construct by adding carefully chosen sets of edges to a low-stretch spanning tree.
They then show that the relationship between the combinatorial properties of a graph and the spectral properties of its Laplacian enables them to
use the Laplacian of each graph in this sequence as a preconditioner for the one preceding it in a recursively applied iterative solver, such as the Preconditioned Chebyshev Method or Preconditioned Conjugate Gradient.

We remark that multigrid methods (see, e.g.,~\cite{MultigridBook}), which are widely used in practice on graphs with sufficiently nice topologies, can be thought of as following a similar multilevel recursively preconditioned iterative framework. Indeed, Koutis \emph{et al.} have an algorithm and implementation based on related techniques that they refer to as ``combinatorial multigrid''~\cite{DBLP:conf/sc/BlellochKMT10}.Even if one does not demand provable running time bounds, we are not aware of any algorithm whose running time empirically scales nearly linearly on large classes of input graphs that does not  roughly follow this general structure.

\subsection{Our Results}

In this paper, we present a new, simple, combinatorial algorithm that solves SDD systems and has a running time of $\runtimeBest$. It uses very little of the machinery that previously appeared to be necessary for a nearly linear time algorithm. It does not require spectral sparsifiers (or variants such as ultra-sparsifiers or incremental sparsifiers), recursive preconditioning, or even the Chebyshev Method or Conjugate Gradient.

To solve an SDD system all our solver requires is a single low-stretch spanning tree\footnote{This can be obtained in nearly-linear time by a simple ball-growing algorithm \cite{AKPW}; the constructions with the best known parameters use a more intricate, but still nearly-linear-time, region growing technique \cite{EEST, ABNlowStretch08, KoutisMP10, AbrahamNeiman12}. In \sectionref{sec:no-spanning-tree} we discuss how even this requirement can be relaxed.} of $G$ (not a recursive collection of subgraphs), and a straightforward data structure.
Given these, the algorithm can be described in a few lines of pseudocode, and its analysis can be made to fit on a single blackboard.

Due to the complexity of previous nearly linear time solvers and the intricate and delicate nature of their analyses, it was necessary to apply them as a black box.
By providing a new, easy-to-understand algorithm, it is our hope that algorithms that use SDD solvers can be improved by ``opening up'' this black box and modifying it to take advantage of specific features of the problem, and that similar techniques can be applied to related problems (e.g., ones with additional constraints or slight deviations from linearity or diagonal dominance).

Because of the lightweight nature of the algorithm and data structure, we expect it to be fast in practice.  Furthermore, the analysis is quite robust, and we believe that the algorithm can be readily adapted to work in multicore, distributed, and even asynchronous settings.   We hope to evaluate our algorithm empirically in a follow-up paper.

\paragraph{Numerical Stability and Run-time Comparisons}
Existing algorithms relied on Preconditioned Chebyshev methods, whose numerical stability is quite difficult to analyze.
At present, the best known results show that they can be implemented with finite-precision arithmetic, but the number of bits of precision required is
$\log \kappa(\lap) \log^c n \log \epsilon^{-1}$, where $\kappa(\lap)$ is the condition number of $\lap$, and $c$ is some possibly large constant~\cite{ST08c}.
The stated running time of the best existing algorithm is $O(m \log n \log \epsilon^{-1})$ (ignoring  $O(\poly(\log\log n))$ terms and using the analysis from~\cite{KMP11} and the best known low-stretch spanning tree algorithm~\cite{AbrahamNeiman12}), but this assumes arbitrary precision arithmetic.  If one analyzes it in the more standard unit-cost RAM model, where one can perform operations only on $O(\log n)$-bit numbers in constant time, this introduces several additional logarithmic factors in the running time.

In contrast, we show that our algorithm is numerically stable and does not require this additional overhead in the bit precision.  As such, our algorithm gives the fastest known algorithm for solving SDD systems in the unit-cost RAM model. Our algorithm approximately solves both SDD systems and the dual electrical flow problem in time $\runtimeBest$. If one allows infinite precision arithmetic to be performed in constant time, this is slower than \cite{KMP11}  by a factor of $O(\log n)$ for the SDD system problem and by a factor of $O(\min\{\log(\epsilon^{-1}),\log n\})$ for the electrical flow problem (ignoring $O(\poly(\log \log n))$ terms).\footnote{To the best of our knowledge, to convert their SDD system solver to an $\epsilon'$-approximate electrical flow solver, one needs to pick $\epsilon=O(\frac{\epsilon'}{n})$.}

\subsection{Overview of our approach}
\label{sec:intro-overview}
Using standard reductions techniques we reduce solving arbitrary SDD systems to solving $\lap \volt = \boundary$, where $\lap \in \R^{n \times n}$ is the Laplacian of a weighted connected graph $G=(V,E,W)$ with $n$ vertices and $m$ edges and $\boundary \in \R^{n}$, see \appendixref{appendix:reductions} for details. To simplify this overview, suppose for now that $\boundary = \vvar{e}_s - \vvar{e}_t$ and $\vvar{e}_s$ and $\vvar{e}_t$ are the unit basis vectors corresponding to two vertices in the graph, $s$ and $t$ respectively.  (The general case is similar and discussed in the following sections.)

Such Laplacian systems $\lap \volt = \boundary$ can be viewed as \emph{electrical flow problems}: each edge $e \in E$ can be viewed as a resistor of resistance $r_e=1/w_e$, where $w_e$ is the weight of this edge, and one unit of electric current needs to be sent from $s$ to $t$. If $\volt$ is a valid solution to $\lap \volt = \boundary$, then the entries of $\volt$ can be viewed as the \emph{electrical potentials} of this system, and the amount of \emph{electric current or flow} $\flow(e)$ on an edge $e=(i,j)$ from $i$ to $j$ is given by $(\volt_i-\volt_j)/r_e$.  The fact that electric flows are induced by vertex potential differences is a crucial property of electrical systems that our algorithm is going to exploit.

While previous solvers worked with the potential vector $\volt$, our algorithm works with the flow vector $\flow$.  Our algorithm begins with any arbitrary unit $s$-$t$ flow (e.g., a path from $s$ to $t$) and maintains its feasibility throughout the execution. If $\flow$ were a valid electrical flow, then it would be induced by some potential vector $\volt \in \rvertvec$ satisfying $\flow(e)=(\volt_i - \volt_j)/r_e$ for all edges $e = (i,j) \in E$, and in particular, for any cycle $C$ in $G$, we would have the potential drop
$
    \sum_{e\in C} \flow(e)r_e=0
$. Exploiting this fact, our entire algorithm consists of repeatedly doing the following:
\begin{itemize} \itemsep -1pt
\item[-] Randomly sample a cycle $C$ from some probability distribution.
\item[-] Compute $\sum_{e\in C}\flow(e) r_e$ and if it is non-zero then add a multiple of $C$ to $\flow$ to make it zero.
\end{itemize}
To turn this into an provably-correct and efficient algorithm we need to do the following:
\begin{itemize}[leftmargin=0.8cm]

\item {\bf Specify the cycle distribution.} The cycles are those found by adding edges to a low-stretch spanning tree. They are sampled proportional to their stretch (See \sectionref{sec:pre-cycles} and \ref{sec:algorithm}).

\item {\bf Bound the number of iterations.} In \sectionref{sec:convergence}, we will show that repeating this process a nearly linear number of times suffices to yield an $\epsilon$-approximate solution.

\item {\bf Implement the iterations efficiently.} Since a cycle may contain a large number of edges, we cannot simply update the flow edge-by-edge. In \sectionref{sec:datastructure}, we give a data structure that allows each iteration to take $O(\log n)$ time and we discuss further implementation details.
\end{itemize}

Interestingly, the algorithm and its analysis also have a geometric interpretation, based on the Kaczmarz method~\cite{Kaczmarz}, which iteratively solves a linear system by alternately picking a constraint and projecting the current solution onto its feasible set. In particular, the algorithm can be recast as an instance of the randomized Kaczmarz method of Strohmer and Vershynin~\cite{SV}. This geometric view is presented in \sectionref{sec:geoview}.

In addition, in \sectionref{sec:improvement} we show how to improve the running time of our algorithm, in \sectionref{sec:stability} we prove the numerical stability of our algorithm, in \sectionref{sec:op-view} we prove that our algorithm can be viewed as a linear operator that spectrally approximates the Laplacian pseudoinverse, and in \sectionref{sec:no-spanning-tree} we discuss how to extend our algorithm to solve SDD systems without using low stretch spanning trees. For the reader who simply wants a complete proof of how to solve SDD systems in nearly-linear time, \sectionref{sec:preliminaries} - \ref{sec:simple_runtime} and \appendixref{appendix:reductions} suffice.

\section{Preliminaries}
\label{sec:preliminaries}

Throughout the remainder of the paper let $G=(V,E,w)$ be a weighted, connected, undirected graph with $n=|V|$ vertices, $m=|E|$ edges and edge weights $w_e>0$.  We think of $w_e$ as the \emph{conductance} of $e$, and we define the \emph{resistance} of $e$ by $r_e \defeq 1/w_e$.
For notational convenience, we fix an orientation of the edges so that for any vertices $a$ and $b$ connected by an edge, exactly one of $(a, b) \in E$ or $(b, a) \in E$ holds.

We make extensive use of the following matrices associated to $G$:
\begin{definition}[Matrix Definitions]
We define the \emph{incidence matrix}, $\incMatrix \in \R^{E \times V}$, the  \emph{resistance matrix}, $\rMatrix \in \R^{E \times E}$,
and the \emph{Laplacian}, $\lap \in \R^{V \times V}$, for all $(a, b), e_1, e_2 \in E$, and $a, b, c \in V$, by
\[
\incMatrix_{(a, b), c} =
        \begin{cases}
        1 & a = c \\
        -1 & b = c \\
        0 & \text{otherwise}
        \end{cases}
\enspace,\enspace
 \rMatrix_{e_1, e_2}
        = \begin{cases}
        r_{e} & e = e_1 = e_2 \\
        0 & \text{otherwise}
        \end{cases}
\enspace,\enspace
    \lap_{a, b}
        = \begin{cases}
            \sum_{\{a, u\} \in E} {w_{a,u}} & a = b \\
            -{w_{a,b}} & \{a, b\} \in E \\
            0 & \text{otherwise}
        \end{cases}
\]
\end{definition}

 For a vector $\flow\in\redgevec$ and an edge $e=(a,b)\in E$, we write $\flow(e)=\flow(a,b)$ for the coordinate of $\flow$ corresponding to $e$, and we adopt the convention that $\flow(b,a)=-\flow(a,b)$.  This allows us to think of $\flow$ as a flow on the graph (not necessarily obeying any conservation constraints) that sends $\flow(e)$ units of flow from $a$ to $b$, and thus $-\flow(e)$ units of flow from $b$ to $a$.

The following facts follow by simple manipulations of the above definitions:

\begin{claim}\label{claim:basicfacts}
For all $\flow \in \redgevec$, $x \in \rvertvec$, $a \in V$ and $(a, b) \in E$:
\begin{itemize}\itemsep -1pt
  \item $\big[\incMatrix^T \flow\,\big]_{a}=         \sum_{(b, a) \in E} \flow(b, a) - \sum_{(a, b) \in E}  \flow(a, b)$ ,
  \item $\lap = \incMatrix^T \rMatrix^{-1} \incMatrix$ ,
  \item $\left[\incMatrix x\right]_{(a, b)} = x(a) - x(b)$ , and
  \item $x^T \lap x = \sum_{(a, b) \in E}
    \frac{\left(x_a - x_b\right)^2}{r_{a,b}}$ .
\end{itemize}
\end{claim}

One can interpret the first assertion in \claimref{claim:basicfacts} as saying that $\incMatrix^T \flow$ is a vector in $\rvertvec$
whose $a$-th coordinate indicates how much flow $\flow$ leaves (or enters, if it is negative) the graph $G$ at vertex $a$.
We say that $\flow$ is a \emph{circulation} if $\incMatrix^T\flow=0$.

\subsection{Electrical Flow}
\label{sec:pre-electrical-flow}
For any vector $\flow\in \redgevec$, we define its \emph{energy} $\renergy{\flow}$ by
\[
    \renergy{\flow}
        \defeq \sum_{e \in E} r_e \flow(e)^2
        = \flow^T \rMatrix \flow
        = \|\flow\|_{\rMatrix}^2  \enspace.
\]
We  fix a \emph{demand vector} $\boundary \in \rvertvec$ and we say a flow $\flow \in \redgevec$ is \emph{feasible (with respect to $\boundary$)}, or that it \emph{meets the demands}, if $\incMatrix^T \flow = \boundary$. Since $G$ is connected it is straightforward to check that there exists a feasible flow with respect to $\boundary$ if and only if $\sum_{v \in V} \boundary(v) = 0$.
\begin{definition} [Electrical Flow]
For a demand vector $\boundary \in \rvertvec$ satisfying $\sum_{a \in V} \boundary(a) = 0$, the \emph{electrical flow satisfying $\boundary$} is the unique minimizer to the following
\begin{equation}\label{eqn:primal}
\optFlow \defeq \argmin_{\flow \in \redgevec \enspace : \enspace \incMatrix^T \flow = \boundary}
    \renergy{\flow} \enspace.
\end{equation}
\end{definition}

This quadratic program  describes a natural physical problem. Given an electric circuit with nodes in $V$, for each undirected edge $e=\{a,b\}\in E$ we connect nodes $a$ and $b$ with a resistor of resistance $r_e$. Next, we fix the amount of current entering and leaving each node and denote this by demand vector $\boundary$. Recalling from physics that the energy of sending $i$ units of current over a resistor of resistance $r$ is $i^2 \cdot r$ the amount of electric current on each resistor is given by $\optFlow$.

The central problem of this paper is to efficiently compute an $\epsilon$-approximate electrical flow.

\begin{definition}[$\epsilon$-Approximate Electrical Flow]
For any $\epsilon \in \rNonNeg$, we say $\flow \in \redgevec$ is an \emph{$\epsilon$-approximate electric flow satisfying $\boundary$} if
$\incMatrix^T \flow = \boundary, \text{ and } \renergy{\flow} \leq (1 + \epsilon) \cdot \renergy{\optFlow} \enspace.$
\end{definition}

\subsubsection{Duality}
\label{sec:pre-duality}

The electric flow problem is dual to solving $\lap \vvar{x} = \boundary$ when $\lap$ is the Laplacian for the same graph $G$. To see this, we note that the Lagrangian dual of \equationref{eqn:primal} is given by
\begin{equation}
\label{eqn:dual}
    \max_{\volt \in \rvertvec} 2\volt^T \boundary -
    \volt^T \lap \volt \enspace.
\end{equation}
For symmetry we define the \emph{(dual) energy} of $\volt \in \rvertvec$ as
$
    \renergyDual{\volt} \defeq 2 \volt^T \boundary -
    \volt^T \lap \volt
$.
Setting the gradient of \equationref{eqn:dual} to $0$ we see that \equationref{eqn:dual} is minimized by $\volt \in \rvertvec$ satisfying
$
    \lap \volt = \boundary \enspace.
$
Let $\lapPseudo$ denote the \emph{Moore-Penrose pseduoinverse} of $\lap$ and let $\optVolt \defeq \lapPseudo \boundary$ denote a particular set of optimal voltages . Since the primal program \equationref{eqn:primal} satisfies Slater's condition, we have strong duality, so for all $\volt \in \rvertvec$
\[
    \renergyDual{\volt} \leq \renergyDual{\optVolt} = \renergy{\optFlow} \enspace.
\]
Therefore, for feasible $\flow \in \redgevec$ and $\volt \in \rvertvec$ the \emph{duality gap}, $\gap(\flow, \volt) \defeq \renergy{\flow} - \renergyDual{\volt}$ is an upper bound on both
$\renergy{\flow} - \renergy{\optFlow}$ and $\renergyDual{\optVolt} - \renergyDual{\volt}$.

In keeping with the electrical circuit interpretation we refer to a candidate dual solution $\volt \in \rvertvec$ as \emph{voltages} or \emph{vertex potentials} and we define $\potentialdrop{\volt}{a, b} \defeq \volt(a) - \volt(b)$ to be the \emph{potential drop} of $\volt$ across $(a, b)$. By the KKT conditions we know that
\[
    \optFlow = \rMatrix^{-1} \incMatrix \optVolt
~~ \text{ i.e. } ~~
    \forall e \in E ~ : ~ \optFlow(e) =
        \frac{\potentialdrop{\optVolt}{e}}{r_e} \enspace.
\]
For $e \in E$ we call $\frac{\potentialdrop{\volt}{e}}{r_e}$
the \emph{flow induced by $\volt$ across $e$} and we call
$\flow(e) r_e$ \emph{the potential
drop induced by $\flow$ across $e$}.

The optimality conditions $\optFlow = \rMatrix^{-1} \incMatrix \optVolt$ can be restated solely in terms of flows in a well known variant of \emph{Kirchoff's Potential Law} (KPL) as follows
\begin{lemma}[KPL]\label{lemma:kpl}
Feasible $\flow \in \redgevec$ is optimal if and only if $\flow^T \rMatrix \circVec = 0$ for all circulations $\circVec \in \redgevec$.
\end{lemma}

\subsection{Spanning Trees and Cycle Space}
\label{sec:pre-cycles}

Let $\tree \subseteq E$ be a spanning tree of $G$ and let us call the edges in $\tree$ the \emph{tree edges} and the edges in $E \setminus \tree$ the \emph{off-tree edges}.
Now, by the fact that $G$ is connected and $T$ spans $G$ we know that for every $a, b \in V$ there is a unique path connecting $a$ and $b$ using only tree edges.

\begin{definition}[Tree Path]
For $a, b \in V$, we define the \emph{tree path} $\treePath{(a,b)} \subseteq V
\times V$ to be the unique path from $a$ to $b$ using edges from $T$.
\footnote{Note that the edges of $\treePath{(a, b)}$ are oriented with respect
to the path, not the natural orientation of $G$.} In
vector form we let $\treePathVec{(a,b)} \in \redgevec$ denote the unique flow
sending $1$ unit of flow from $a$ to $b$, that is nonzero only on $T$.
\end{definition}

For the off-tree edges we similarly define \emph{tree cycles}.

\begin{definition}[Tree Cycle]
For $(a, b) \in \offtreeEdgeSet$, we define the \emph{tree cycle}
$\treeCycle{(a,b)} \defeq \{(a,b)\} \cup \treePath{(b, a)}$ to be the unique cycle consisting of
edge $(a, b)$ and $\treePath{(b, a)}$. In vector form we let
$\treeCycleVec{(a,b)}$ denote the unique circulation sending $1$ unit of flow on
$\treeCycle{(a, b)}$.
\end{definition}

\subsubsection{Cycle Space}

The tree cycles form a complete characterization of circulations in a graph. The set of all circulations
$
    \{\circVec \in \redgevec ~ | \incMatrix^T \circVec = 0\}
$
is a well-known subspace called \emph{cycle space} \cite{Bollobas98} and the tree cycles
$
    \{\treeCycleVec{e} ~ | ~ e \in \offtreeEdgeSet\}
$
form a basis. This yields an even more succinct description of the KPL optimality condition (\lemmaref{lemma:kpl}). A feasible $\flow \in \redgevec$ is optimal if and only if $\flow^T R \circVec_e = 0$ for all $e \in \offtreeEdgeSet$.

We can think of each tree cycle $\treeCycle{e}$ as a long resistor consisting of its edges in series with total resistance $\sum_{e \in \treeCycle{e}} r_e$ and flow induced potential drop of $\sum_{e \in \treeCycle{e}} \flow(e) r_e$. KPL optimality then states that $\flow \in \redgevec$ is optimal if and only if the potential drop across each of these resistors is 0. Here we define two key quantities relevant to this view.

\begin{definition}[Cycle Quantities]
For $e \in \offtreeEdgeSet$ and $\flow \in \redgevec$ the \emph{resistance of $\treeCycle{e}$}, $\cycleResistance{e}$, and the \emph{flow induced potential across $\treeCycle{e}$}, $\cyclePotential{e}{\flow}$, are given by
\[
\cycleResistance{e} \defeq \sum_{e' \in \treeCycle{e}} r_{e'} = \treeCycleVec{e}^T \rMatrix \treeCycleVec{e}
\enspace\enspace \text{and} \enspace\enspace
\cyclePotential{e}{\flow}
\defeq \sum_{e \in \treeCycle{e}} r_e \flow(e)  = \flow^T \rMatrix \treeCycleVec{e}
\enspace.
\]
\end{definition}

\subsubsection{Low-Stretch Spanning Trees}

Starting with a feasible flow, our algorithm computes an approximate electrical flow by fixing violations of KPL on randomly sampled tree cycles. How well this algorithm performs is determined by how well the resistances of the off-tree edges are approximated by their corresponding cycle resistances. This quantity, which we refer to as as the \emph{tree condition number}, is in fact a certain condition number of a matrix whose rows are properly normalized instances of $\treeCycleVec{e}$ (see \sectionref{sec:geoview}).
\begin{definition}[Tree Condition Number]
The \emph{tree condition number} of spanning tree $\tree$ is given by
$
    \treeCondition(\tree)
        \defeq \sum_{e \in \offtreeEdgeSet}
            \frac{\cycleResistance{e}}{r_e} \enspace,
$
and we abbreviate it as $\treeCondition$ when the underlying tree $\tree$ is clear from context.
\end{definition}
\noindent This is closely related to a common quantity associated with a spanning tree called \emph{stretch}.
\begin{definition}[Stretch]
The \emph{stretch} of $e \in E$, $\stretchEdge{e}$, and the \emph{total stretch} of $\tree$, $\stretchTotal{\tree}$, are
\[
    \stretchEdge{e}
        \defeq \frac{\sum_{e' \in \treePath{e}} r_{e'}}{r_e}
    \enspace \text{ and } \enspace
        \stretchTotal{\tree} \defeq \sum_{e \in E} \stretchEdge{e}
    \enspace.
\]
\end{definition}

Since $\cycleResistance{e} = r_e \cdot (1 + \stretchEdge{e})$ we see that these quantities are related by $\treeCondition(T) = \stretchTotal{\tree} + m - 2n + 2$.

Efficient algorithms for computing spanning trees with low total or average stretch, i.e. \emph{low-stretch spanning trees}, have found numerous applications \cite{AKPW,EEST} and all previous nearly-linear-time SDD-system solvers \cite{ST04, ST08c, KoutisMP10, KMP11}, including the most efficient form of the SDD solver presented in this paper, make use of such trees. There have been multiple breakthroughs in the efficient construction of \emph{low-stretch spanning trees} \cite{AKPW, EEST, ABNlowStretch08, KoutisMP10, AbrahamNeiman12} and the latest such result is used in this paper and stated below.

\begin{theorem}[\cite{AbrahamNeiman12}]
\label{thm:low-stretch}
In $O(m \log n \log \log n)$ time we can compute a spanning tree $\tree$ with total stretch $\stretchTotal{\tree} = O(m \log n \log \log n)$.
\end{theorem}

\section{A Simple Nearly-Linear Time SDD Solver}
\label{sec:algorithm}

Given a SDD system $\mvar{A} \vvar{x} = \vvar{b}$ we wish to efficiently compute $\vvar{x}$ such that $\|\vvar{x} - \pseudo{\mvar{A}} \vvar{b} \|_{\mvar{A}} \leq \epsilon \|\pseudo{\mvar{A}} \vvar{b}\|_{A}$. Using standard reduction techniques (detailed in \appendixref{appendix:reductions}) we can reduce solving such SDD systems to solving Laplacian systems corresponding to connected graphs without a loss in asymptotic run time. Therefore it suffices to solve $\lap \volt = \boundary$ in nearly-linear time when $\lap$ is the Laplacian matrix for some connected graph $G$. Here we provide an algorithm that both solves such systems and computes the corresponding $\epsilon$-approximate electric flows in $\runtimeSimple$ time.

\subsection{A Simple Iterative Approach}

Our algorithm focuses on the electric flow problem. First, we compute a low stretch spanning tree, $\tree$, and a crude initial feasible $\flowInitial \in \redgevec$ taken to be the unique $\flowInitial$ that meets the demands and is nonzero only on tree edges. Next, for a fixed number of iterations, $\nIter$, we perform simple iterative steps, referred to as \emph{cycle updates}, in which we compute a new feasible $\flow_{i} \in \redgevec$ from the previous feasible $\flow_{i - 1} \in \redgevec$ while attempting to decrease energy. Each iteration, $i$, consists of sampling an $e \in \offtreeEdgeSet$ proportional to $\frac{\cycleResistance{e}}{r_e}$, checking if $\flow_{i - 1}$ violates KPL on $\treeCycle{e}$ (i.e. $\cyclePotential{e}{\flow_{i}} \neq 0$) and adding a multiple of $\treeCycleVec{e}$ to make KPL hold on $\treeCycle{e}$ (i.e. $\flow_{i} = \flow_{i - 1} - \frac{\cyclePotential{e}{\flow_{i - 1}}}{\cycleResistance{e}} \treeCycleVec{e}$). Since, $\incMatrix^T \treeCycleVec{e} = 0$, this operation preserves feasibility. We show that in expectation $\flow_{\nIter}$ is an $\epsilon$-approximate electrical flow.

To solve $\lap \volt = \boundary$ we show how to use an $\epsilon$-approximate electrical flow to derive a candidate solution to $\lap \volt = \boundary$ of comparable quality. In particular, we use the fact that a vector $\flow \in \redgevec$ and a spanning tree $\tree$ induce a natural set of voltages, $\volt \in \rvertvec$ which we call the \emph{tree induced voltages}.

\begin{definition}[Tree Induced Voltages]
\label{def:tree-induced-voltages}
For $\flow \in \redgevec$ and an arbitrary (but fixed) $s \in V$,\footnote{We are primarily concerned with difference between potentials which are invariant under the choice $s \in V$.} we define the \emph{tree induced voltages} $\volt \in \rvertvec$ by
$
\volt(a)
    \defeq \sum_{e \in \treePath{(a, s)}} \flow(e) r_e
$
for $\forall a \in V$.
\end{definition}

Our algorithm simply returns the tree induced voltages for $\flow_{\nIter}$, denoted $\volt_{\nIter}$, as the approximate solution to $\lap \volt = \boundary$. The full pseudocode for the algorithm is given in \algorithmref{algm:simplealgorithm}.

\begin{algorithm}
\SetAlgoLined
\SetKwInOut{Input}{Input}
\SetKwInOut{Output}{Output}
\Input{$G = (V, E, r)$, $\boundary \in \rvertvec$, $\epsilon \in \rPos$}
\Output{$\flow \in \redgevec$ and $\volt \in \rvertvec$}
\BlankLine
$T := $ low-stretch spanning tree of $G$\;
$\flow_0 := $ unique flow on $T$ such that $\incMatrix^T \flow_0 = \boundary$\;
$\edgeSampleProb{e} := \frac{1}{\treeCondition(\tree)} \cdot \frac{\cycleResistance{e}}{r_e}$ for all
$e \in \offtreeEdgeSet$ \;
$\nIter = \big\lceil\treeCondition \log \big(\frac{\stretchTotal{\tree}
    \cdot \treeCondition(\tree)}{\epsilon}\big)\big\rceil$\;
\For{$i = 1$ to $\nIter$}
{
    Pick random $e_i \in \offtreeEdgeSet$ by probability distribution
    $\sampleProbVec$ \;
    $\flow_i = \flow_{i - 1} - \frac{\cyclePotential{e}{\flow_{i - 1}}}
        {\cycleResistance{e}}
    \treeCycleVec{e}$ \;
}
\Return{$\flow_{\nIter}$ and its tree induced voltages $\volt_{\nIter}$}
\caption{\label{algm:simplealgorithm}$\simpleSolver$}
\end{algorithm}

\subsection{Algorithm Guarantees}

In the next few sections we prove that $\simpleSolver$ both computes an $\epsilon$-approximate electric flow and solves the corresponding Laplacian system in nearly linear time:
\begin{theorem}[$\simpleSolver$]
\label{thm:simple_algorithm}
The output of $\simpleSolver$ satisfies
\footnote{Although the theorem is stated as an expected guarantee on $\flow$ and $\volt$, one can easily use Markov bound and Chernoff bound to provide a probabilistic but exact guarantee.}
\[
\E\big[\renergy{\flow}\big] \leq (1 + \epsilon) \cdot \renergy{\optFlow}
\enspace \text{ and } \enspace
\E\big\|\volt - \lapPseudo \boundary\big\|_{\lap}
    \leq \sqrt{\epsilon} \cdot \big\|\lapPseudo \boundary\big\|_{\lap}
\]
and $\simpleSolver$ can be implemented to run in time $\runtimeSimple$.
\end{theorem}

By construction $\simpleSolver$ outputs a feasible flow and, by choosing $\tree$ to be a low stretch spanning tree with properties guaranteed by \theoremref{thm:low-stretch}, we know that the number of iterations of simple solver is bounded by $O(m \log n \log \log n \log (\epsilon^{-1} n))$. However, in order to prove the theorem we still need to show that (1) each iteration makes significant progress, (2) each iteration can be implemented efficiently, and (3) the starting flow and final voltages are appropriately related to $\renergy{\optFlow}$. In particular we show:
\begin{enumerate}
  \item Each iteration of $\simpleSolver$ decreases the energy of the current flow by at least an expected $\left(1 - \frac{1}{\treeCondition}\right)$ fraction of the energy distance to optimality (see \sectionref{sec:convergence}).
  \item Each iteration of $\simpleSolver$ can be implemented to take $O(\log n)$ time (see \sectionref{sec:datastructure}).
\footnote{Note that a na\"{\i}ve implementation of cycle updates does not necessarily run in sublinear time. In particular, updating
$\flow_i := \flow_{i - 1} - \frac{\cyclePotential{e}{\flow_{i - 1}}}{\cycleResistance{e}}
\treeCycleVec{e}$
by walking $\treeCycle{e}$ and updating flow values one by one may take more than (even amortized) sublinear time, even though $\tree$ may be of low total stretch. Since $\stretchTotal{\tree}$ is defined with respect to cycle resistances but not with respect to the number of edges in these cycles, it is possible to have a low-stretch tree where each tree cycle still has $\Omega(|V|)$ edges on it.
Furthermore, even if all edges have resistances $1$ and therefore the average number of edges in a tree cycle is $\tilde{O}(\log n)$, since $\simpleSolver$ samples off-tree edges with higher stretch with higher probabilities, the expected number of edges in a tree cycle may still be $\Omega(|V|)$.}
  \item The energy of $\flow_0$ is sufficiently bounded, the quality of tree voltages is sufficiently bounded, and all other parts of the algorithm can be implemented efficiently  (see \sectionref{sec:simple_runtime}).
\end{enumerate}

\section{Convergence Rate Analysis}
\label{sec:convergence}

In this section we analyze the convergence rate of $\simpleSolver$. The central result is as follows.

\begin{theorem}[Convergence]
\label{thm:convergence}
Each iteration $i$ of $\simpleSolver$ computes feasible $\flow_i \in \redgevec$ such that
\[
    \E\big[\renergy{\flow_i}\big] - \renergy{\optFlow} \leq \Big(1 - \frac{1}{\treeCondition}\Big)^{i}
    \left(\renergy{\flowInitial} - \renergy{\optFlow}\right)
    \enspace.
\]
\end{theorem}

Our proof is divided into three steps. In \sectionref{sec:convergence:step_absolute} we analyze the energy gain of a single algorithm iteration, in \sectionref{sec:convergence:step_relative} we bound the distance to optimality in a single algorithm iteration, and in \sectionref{sec:convergence:total} we connect these to prove the theorem.

\subsection{Cycle Update Progress}
\label{sec:convergence:step_absolute}

For feasible $\flow \in \redgevec,$ we can decrease $\renergy{\flow}$ while maintaining feasibility, by adding a multiple of a circulation $\circVec \in \redgevec$ to $\flow$. In fact, we can easily optimize to pick the best multiple.

\begin{lemma}[Energy Improvement]
For $\flow \in \redgevec$, $\circVec \in \redgevec$, and $\alpha^* = - \frac{\flow^T \rMatrix \circVec}
    {\vvar{c}^T \rMatrix \circVec} \in \R$ we have
\[
    \argmin_{\alpha \in \R} ~ \renergy{\flow + \alpha \circVec}
    = - ~ \frac{\flow^T \rMatrix \circVec}{\circVec^T \rMatrix \circVec}
\enspace \enspace \text{ and } \enspace \enspace
    \renergy{\flow + \alpha^* \circVec} - \renergy{\flow}
    = - ~ \frac{\big(\flow^T \rMatrix \circVec\big)^2}
        {\circVec^T \rMatrix \circVec}
\enspace.
\]
\end{lemma}

\begin{proof}
By definition of energy
$\renergy{\flow + \alpha \circVec}
    = (\flow + \alpha \circVec)^T \rMatrix (\flow + \alpha \circVec)
    =
        \flow^T \rMatrix \flow
        + 2 \alpha \flow^T \rMatrix \circVec
        + \alpha^2 \circVec^T \rMatrix \circVec \enspace.
$
Setting the derivative with respect to $\alpha$ to 0 and substituting in $\alpha = \alpha^*$ yields the results.
\end{proof}

In the special case where $\circVec = \treeCycleVec{e}$ is a tree cycle for some off-tree edge $e \in \offtreeEdgeSet$, since $\cycleResistance{e} = \treeCycleVec{e}^T \rMatrix \treeCycleVec{e}$ and  $\cyclePotential{e}{\flow} = \flow^T \rMatrix \treeCycleVec{e}$, this procedure is precisely the iterative step of $\simpleSolver$, i.e. a cycle update. The following lemma follows immediately and states that the energy decrease of a cycle update is exactly the energy of a resistor with resistance $\cycleResistance{e}$ and potential drop $\cyclePotential{e}{\flow}$.

\begin{lemma}[Cycle Update]
\label{lemma:cycle-update}
For feasible $\flow \in \redgevec$ and $e \in \offtreeEdgeSet$ we have
\[
\renergyFull{\flow - \frac{\cyclePotential{e}{\flow}}{\cycleResistance{e}} \treeCycleVec{e}}
- \renergy{\flow}
= - \frac{\cyclePotential{e}{\flow}^2}{\cycleResistance{e}} \enspace.
\]
\end{lemma}

\subsection{Distance to Optimality}
\label{sec:convergence:step_relative}

To bound how far $\flow_i$ in $\simpleSolver$ is from optimality we use that the duality gap between $\flow_i$ and its tree induced voltages $\volt_i$ is an upper bound on this distance. Here we derive a simple expression for this quantity in terms of cycle potentials.

\begin{lemma}[Tree Gap]

\label{lemma:tree_gap_formula}
For feasible $\flow \in \redgevec$ and tree induced voltages $\volt \in \rvertvec$ we have
\[
    \gap(\flow, \volt) =
        \sum_{e \in E \setminus \tree}
            \frac{\cyclePotential{e}{\flow}^2}{r_e} \enspace.
\]
\end{lemma}

\begin{proof}
By definition of primal and dual energy we have
$
    \gap(\flow, \volt)
        = \flow^T \rMatrix \flow
            - (2\volt^T \boundary - \volt^T \lap \volt)
$.
Using that $\incMatrix^T \flow = \boundary$ and
$\lap = \incMatrix^T \rMatrix^{-1} \incMatrix$ we get
\[
\gap(\flow, \volt)
    = \flow^T \rMatrix \flow
        - 2\volt^T \incMatrix^T \flow
        - \volt^T \incMatrix^T \rMatrix^{-1} \incMatrix \volt
    = \left(\rMatrix \flow - \incMatrix \volt\right)^T
        \rMatrix^{-1}
      \left(\rMatrix \flow - \incMatrix \volt\right)\enspace.
\]
Therefore $\gap(\flow, \volt) = \sum_{e \in E} \frac{1}{r_e} \left(\flow(e) r_e  - \potentialdrop{\volt}{e}\right)^2$. However, by the
uniqueness of tree paths, the antisymmetry of $\flow \in \redgevec$, and the definition of tree voltages we have
\[
\forall a, b \in V \enspace : \enspace
\potentialdrop{\volt}{a,b}
= \volt(a) - \volt(b)
= \sum_{e \in \treePath{as}} \flow(e) r_e
        + \sum_{e \in \treePath{sb}} \flow(e) r_e
= \sum_{e \in \treePath{ab}} \flow(e) r_e
\enspace.
\]
Therefore, $e \in \tree \Rightarrow \flow(e) r_e - \potentialdrop{\volt}{e} = 0$ and $e \in \offtreeEdgeSet \Rightarrow \flow(e) r_e - \potentialdrop{\volt}{e} = \cyclePotential{e}{\flow}$.
\end{proof}

\subsection{Convergence Proof}
\label{sec:convergence:total}

Here we connect the energy decrease of a cycle update given by \lemmaref{lemma:cycle-update} and the duality gap formula given by \lemmaref{lemma:tree_gap_formula} to bound the the convergence of $\simpleSolver$. Throughout this section we use $\volt_i$ to denote the tree induced voltages for flow $\flow_i$.

First, we show that in expectation each iteration $i$ of $\simpleSolver$ decreases $\renergy{\flow_{i - 1}}$ by a $\frac{1}{\treeCondition}$ fraction of the duality gap.

\begin{lemma}[Expected Progress] For iteration $i$ of $\simpleSolver$ we have
\label{lemma:expected_progress}
\[
    \E\left[\renergy{\flow_{i}} - \renergy{\flow_{i-1}} \middle| \gap(\flow_{i-1},\volt_{i-1})\right]
        = -\frac{\gap(\flow_{i-1},\volt_{i-1})}{\treeCondition} \enspace.
\]
\end{lemma}

\begin{proof}
In each iteration $i$, $\simpleSolver$ picks a random $e_i \in \offtreeEdgeSet$ with probability $p_{e_i}$ and adds a multiple of $\treeCycleVec{e_i}$ which by
\lemmaref{lemma:cycle-update} decreases the energy by $\frac{\cyclePotential{e}{\flow_{i-1}}^2}{\cycleResistance{e}}$. Therefore
\[
   \E\left[\renergy{\flow_{i}} - \renergy{\flow_{i-1}} \middle| \gap(\flow_{i-1},\volt_{i-1})\right]
= \E\left[\sum_{e \in \offtreeEdgeSet} p_e \left( \frac{-
        \cyclePotential{e}{\flow_{i-1}}^2}{\cycleResistance{e}} \right) \middle| \gap(\flow_{i-1},\volt_{i-1}) \right]
\]
Using that $p_e = \frac{1}{\treeCondition} \cdot \frac{\cycleResistance{e}}{r_e}$ and applying \lemmaref{lemma:tree_gap_formula} yields the result.
\end{proof}

Next, we show that this implies that each iteration decreases the expected energy difference between the current flow and the optimal flow by a multiplicative $\left(1 - \frac{1}{\treeCondition}\right)$.

\begin{lemma}[Convergence Rate]
\label{lemma:convergence_rate}
For all $i \geq 0$ let random variable $D_i\defeq \renergy{\flow_i} - \renergy{\optFlow}$. Then for all iterations $i \geq 1$ we have
$
  \E[D_{i}] \leq \Big(1 - \frac{1}{\treeCondition}\Big) \E[D_{i-1}] \enspace.
$
\end{lemma}

\begin{proof}
Since in each iteration of $\simpleSolver$ one of a finite number of edges is chosen, clearly $D_i$ is a discrete random variable and by law of total expectation we
have
\[
\E[D_{i}]
    = \sum_c \E[D_{i} | D_{i-1} = c] \Pr[D_{i-1} = c]
\]
However, we know $D_{i-1} \leq \gap(\flow_{i-1}, \volt_{i-1})$ so by \lemmaref{lemma:expected_progress},
$
    \E\left[D_{i} | D_{i-1} = c\right] \leq c - \frac{c}{\treeCondition}
$.
Therefore:
\[
\E[D_{i}] \leq
    \sum_c
    \left[
        \Big(1 - \frac{1}{\tau}\Big) c \cdot \Pr[D_{i-1} = c]
    \right]
= \left(1 - \frac{1}{\tau}\right) \E[D_{i - 1}] \enspace.
\]
\end{proof}

Finally, by induction on \lemmaref{lemma:convergence_rate} and the definition of $D_i$ we prove \theoremref{thm:convergence}.

\section{Cycle Update Data Structure}
\label{sec:datastructure}

In this section we show how to implement each iteration (i.e. cycle update) of $\simpleSolver$ in $O(\log n)$ time. In \sectionref{sec:implementation:ds_problem} we formulate this as a data structure problem, in \sectionref{sec:implementation:ds_rec_slution} we present a recursive solution, and in \sectionref{sec:implementation:ds_vec_slution} we provide a linear algebraic interpretation that may be useful in both theory and practice.

\subsection{The Data Structure Problem}
\label{sec:implementation:ds_problem}

In each iteration $i$ of $\simpleSolver$ we pick a random $(a, b) \in \offtreeEdgeSet$ and for feasible $\flow \in \redgevec$ compute
\[
    \alpha^*
        = \frac{\cyclePotential{(a,b)}{\flow}}{\cycleResistance{(a, b)}}
        = \frac{\flow(a, b) r_{a,b} -
            (\volt(a) - \volt(b))}{\cycleResistance{(a,b)}}
\enspace \text{ where } \enspace
    \volt(a) = \sum_{e \in P_{(a, s)}} \flow(e) r_e
\]
and $s$ is an arbitrary fixed vertex in $V$ which we refer to as the \emph{root}. Then $\alpha^*$
is added to the flow on every edge in the tree cycle $\treeCycle{(a, b)}$. By the antisymmetry of flows and the uniqueness of tree paths, this update is equivalent to (1) adding $\alpha^*$ to the flow on edge $(a, b)$, (2) adding $-\alpha^*$ to the flow on every edge in $\treePath{(s, b)}$ and then (3) adding $\alpha^*$ to the flow on every edge in $\treePath{(s, a)}$.

Therefore, to implement cycle updates it suffices to store the flow on off-tree edges in an array and implement a data structure that supports the following operations.
\begin{itemize}\itemsep -1pt
\item $\dsOpInit(\tree, \dsVarSource \in V)$: initialize a data structure $D$ given tree $\tree$ and root $s \in V$.
\item $\dsOpQuery_D(a \in V)$: return $\volt(a) = \sum_{e \in P_{(s, a)}} \flow(e) r_e$.
\item $\dsOpUpdate_D(a \in V, \alpha \in \R)$: set $\flow(e) := \flow(e) + \alpha$ for all $e \in \treePath{(s,a)}$.
\end{itemize}

\subsection{Recursive Solution}
\label{sec:implementation:ds_rec_slution}

While one could solve this data structure problem with a slight modification of link-cut trees~\cite{SleatorTarjan83}, that dynamic data structure is overqualified for our purpose. In contrast to the problems for which link-cut trees were originally designed, in our setting the tree is static, so that much of the sophistication of link-cut trees may be unnecessary. Here we develop a very simple separator decomposition tree based structure that provides worst-case $O(\log n)$ operations that we believe sheds more light on the electric flow problem and may be useful in practice.

Our solution is based on the well known fact that every tree has a good vertex separator, tracing back to Jordan in 1869~\cite{Jordan1869}.

\begin{lemma}[Tree Vertex Separator]
\label{lemma:ds:tree_separator}
For a spanning tree $\tree$ rooted at $s$ with $n\geq 2$ vertices, in $O(n)$ time we can compute
\[
    (\dsVarSeperator, \tree_0, \ldots, \tree_k ) = \dsOpTreeDecompose(\tree) \enspace,
\]
such that the removal of $\dsVarSeperator\in V$ (which might equal to $s$) disconnects $\tree$ into subtrees $\tree_0, \ldots, \tree_k$, where $\tree_0$ is rooted at $s$ and contains $\dsVarSeperator$ as a leaf, while other $\tree_i$'s are rooted at $\dsVarSeperator$. Furthermore, each $\tree_i$ has at most $n/2+1$ vertices.
\end{lemma}
\begin{proof}
\begin{figure*}[hbpt!]
\centering
\hspace{-5mm}
\subfloat[][]{\label{fig:decomp1}\includegraphics[width=0.2\textwidth]{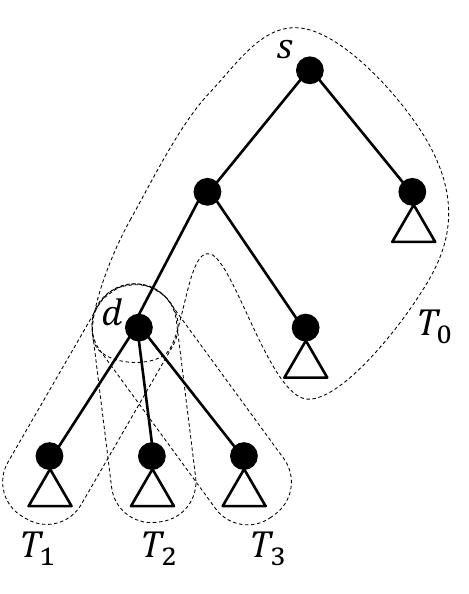}}
\hspace{3mm}
\subfloat[][]{\label{fig:decomp2}\includegraphics[width=0.12\textwidth]{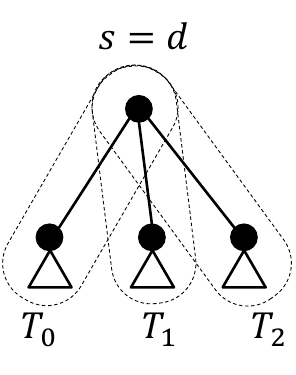}\vspace{2cm}}
\hspace{5mm}
\subfloat[][]{\label{fig:datastructure1}\includegraphics[width=0.29\textwidth]{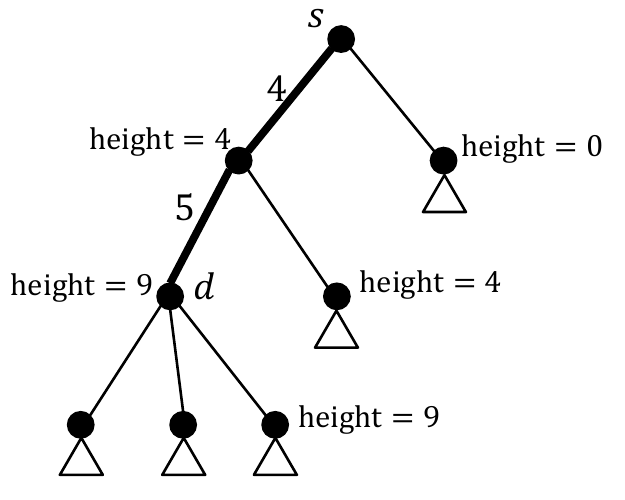}}
\hspace{-5mm}
\subfloat[][]{\label{fig:datastructure2}\includegraphics[width=0.27\textwidth]{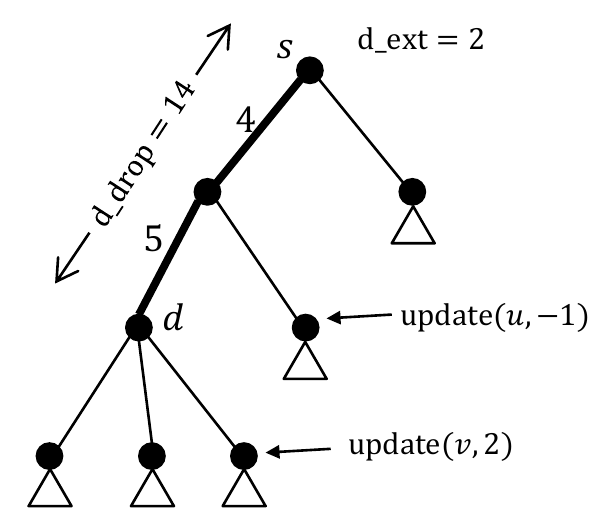}}
\hspace{-5mm}
\caption{(a) and (b) are examples of tree decompositions when $\dsVarSource\neq \dsVarSeperator$ and $\dsVarSource= \dsVarSeperator$, and\\ \hspace{-1.6cm} (c) and (d) are examples on computing $\dsVarHeight$, $\dsVarExt$ and $\dsVarTotal$.}
\label{fig:do-not-use-this}
\end{figure*}
We start at root $s$, and at each step follow an edge to a child whose corresponding subtree is the largest among all children. We continue this procedure until at some vertex $\dsVarSeperator$ (which might be $s$ itself) the sizes of all its subtrees have no more than $\frac{n}{2}$ vertices. This vertex $\dsVarSeperator$ is the desired vertex separator. Now if $s\neq \dsVarSeperator$ we let $\tree_0$ be the subtree above $\dsVarSeperator$ rooted at $s$ with $\dsVarSeperator$ being the leaf, and let each $\tree_i$ be a sub-tree below $\dsVarSeperator$ and rooted at $\dsVarSeperator$, see \figureref{fig:decomp1}; or when $s=\dsVarSeperator$ we let $T_0,\dots,T_k$ each denote a subtree below $\dsVarSeperator$ and rooted at $\dsVarSeperator$, see \figureref{fig:decomp2}. \footnote{For technical purposes, when $V$ has only two vertices with $\dsVarSource$ being the root, we always define $\dsVarSeperator \in V$ such that $\dsVarSeperator \neq \dsVarSource$ and define $\dsOpTreeDecompose(\tree)$ to return no sub-trees, i.e., $k=-1$. By orienting the tree away from $s$ and precomputing the number of descendents of each node we can easily implement this in $O(n)$ time.} By pre-computing subtree sizes $\dsOpTreeDecompose(\tree)$ runs in $O(n)$ time.
\end{proof}

Applying this lemma recursively induces a separator decomposition tree from which we build our data structure. To execute $\dsOpInit(\tree, \dsVarSource)$ we compute a vertex separator $(\dsVarSeperator,\tree_0,\dots,\tree_k)$ and recurse separately on each subtree until we are left with a tree of only two vertices. In addition, we precompute the total weight of the intersection of the root to $\dsVarSeperator$ path $\treePath{(\dsVarSource,\dsVarSeperator)}$ and the root to $a$ path $\treePath{(\dsVarSource,a)}$, denoted by $\dsVarHeight(a)$ for every $a \in V$, by taking a walk over the tree. See \figureref{fig:datastructure1} and \algorithmref{algm:ds:init_recursive} for details.

With this decomposition we can support $\dsOpQuery$ and $\dsOpUpdate$ by maintaining just $2$ values about the paths $\treePath{(s,\dsVarSeperator)}$. For each subtree we maintain the following variables:
\begin{itemize}\itemsep -1pt
\item $\dsVarTotal$, the total potential drop induced on the the path $\treePath{(s,\dsVarSeperator)}$, and
\item $\dsVarExt$, the total amount of flow that has been updated using the entire $\treePath{(s,\dsVarSeperator)}$ path, i.e. the contribution to $\treePath{(s,\dsVarSeperator)}$ from vertices beyond $\dsVarSeperator$.
\end{itemize}
It is easy to see that these invariants can be maintained recursively and that they suffice to answer queries and updates efficiently (see \algorithmref{algm:ds:query_recursive} and \algorithmref{algm:ds:update_recursive} for the implementations and \figureref{fig:datastructure2} for an illustration). Furthermore, because the sizes of trees at least half at each recursion, $\dsOpInit$ takes $O(n \log n)$ while $\dsOpQuery$ and $\dsOpUpdate$ each takes $O(\log n)$ time.

\newcommand{\algRecursiveInit}{
\begin{algorithm}[H]
\SetAlgoLined
$\dsVarExt := 0$, $\dsVarTotal := 0$ \;
$(\dsVarSeperator, T_0, \ldots, T_k) \dsAssign \dsOpTreeDecompose(\tree)$\;
$\forall a \in V : $
    $\dsVarHeight(a) \dsAssign \sum_{e \in \treePath{(s, a)}
        \cap
        \treePath{(s, \dsVarSeperator)}} r_e$ \;
\If{$|V| > 2$}
{
$\forall i \in \{0,1,\dots,k\} : D_i \dsAssign
    \dsOpInit(\tree_i, \dsVarSeperator) $\;}
\caption{\label{algm:ds:init_recursive}Recursive $\dsOpInit(\tree, s \in V)$}
\end{algorithm}
}

\newcommand{\algRecursiveQuery}{
\begin{algorithm}[H]
\SetAlgoLined
\lIf{$a = \dsVarSeperator$}
   {\Return \dsVarTotal} \;
\lElseIf{$|V|=2$}
   { \Return 0\; }
\uElseIf{$a \in T_0$}
   {\Return $\dsVarExt\cdot\dsVarHeight(a) + \dsOpQuery_{D_0}(a)$ \label{prog:1}\; }
\Else
   {
   let $\tree_i$ be unique tree containing $a$\;
   \Return $\dsVarTotal + \dsOpQuery_{D_i}(a)$ \;
   }
\caption{\label{algm:ds:query_recursive}Recursive $\dsOpQuery(a \in V)$}
\end{algorithm}
}

\newcommand{\algRecursiveUpdate}{
\begin{algorithm}[H]
\SetAlgoLined
$\dsVarTotal := \dsVarTotal + \alpha \cdot \dsVarHeight(a)$ \;
\lIf{$|V|=2$}{\Return \;}
\lIf{$a \notin \tree_0$}
{$\dsVarExt := \dsVarExt + \alpha$\;}
\If{$a \neq \dsVarSeperator$}
{
let $\tree_i$ be unique tree containing $a$\;
$\dsOpUpdate_{D_i}(a, \alpha)$\;
}
\caption{\label{algm:ds:update_recursive}Recursive $\dsOpUpdate(a \in V, \alpha \in \R)$}
\end{algorithm}
}

\begin{figure}[h]
\centering
\begin{minipage}[c]{0.48\textwidth}
\algRecursiveInit{}
\end{minipage}\\
\begin{minipage}[c]{0.48\textwidth}
\algRecursiveQuery{}
\end{minipage}
\begin{minipage}[c]{0.48\textwidth}
\algRecursiveUpdate{}
\end{minipage}
\end{figure}

\subsection{Linear Algebra Solution}
\label{sec:implementation:ds_vec_slution}

Here we unpack the recursion in the previous section to provide a linear algebraic view of our data structure that may be useful in both theory and practice. We show that the $\dsOpInit$ procedure in the previous section computes for all $a \in V$ a query vector $\dsVectorQuery{a}$ and an update vector $\dsVectorUpdate{a}$ each of support size $O(\log |V|)$, such that the entire state of the data structure is an $O(|V|)$ dimensional vector $\vvar{x}$ initialized to $\vvar{0}$ allowing $\dsOpQuery$ and $\dsOpUpdate$ to be as simple as the following,
\begin{eqnarray}
    \dsOpQuery(a) &:
        &\texttt{return } \dsVectorQuery{a} \cdot \vvar{x} \nonumber\\
    \dsOpUpdate(a, \alpha)&: &\vvar{x} := \vvar{x} + \alpha \dsVectorUpdate{a} \enspace. \label{eqn:implementation-vec}
\end{eqnarray}
To see this, first note that $\dsOpInit$ recursively creates a total of $N=O(|V|)$ subtrees $\tree_1, \ldots, \tree_N$. Letting $\dsVarExt_i$, $\dsVarTotal_i$, $\dsVarHeight_i$ denote $\dsVarExt$, $\dsVarTotal$, and $\dsVarHeight$ associated with tree $\tree_i$, the state of the data structure is completely represented by a vector $\vvar{x} \in \dsVectorSpace$ defined as follows
\[
\vvar{x}_{c, i} \defeq
    \begin{cases}
        \dsVarExt_i & \hbox{if $c = e$;} \\
        \dsVarTotal_i & \hbox{if $c = d$.} \\
    \end{cases}
\]
It is easy to see that given vertex $a \in V$, in order for $\dsOpQuery(a)$ or $\dsOpUpdate(a, \alpha)$ to affect tree $T_i$ in the decomposition, it is necessary that $a$ is a vertex in $T_i$ and $a$ is not the root $s_i$ of $T_i$. Accordingly, we let $T(a) \defeq \{i \in [N] ~ | ~ a \in \tree_i \text{ and } a \neq \dsVarSource_i\}$ be the indices of such trees that $a$ may affect.

Next, in order to fully specify the query vector $\dsVectorQuery{a}$ and the update vector $\dsVectorUpdate{a}$, we further refine $T(a)$ by defining $T_0(a) \defeq \{i \in T(a) ~ | ~ a \text{ is in $T_0$ of } \dsOpTreeDecompose(T_i)\}$ and $T_{+}(a) \defeq T(a) \setminus T_0(a)$. Then, one can carefully check that the following definitions of $\dsVectorQuery{a}, \dsVectorUpdate{a} \in \dsVectorSpace$ along with their $\dsOpQuery$ and $\dsOpUpdate$ in \equationref{eqn:implementation-vec}, exactly coincide with our recursive definitions in \algorithmref{algm:ds:query_recursive} and \algorithmref{algm:ds:update_recursive} respectively.
\[
        \dsVectorQuery{a}_{(c, i)}
        \defeq
    \begin{cases}
        \dsVarHeight_i(a) & c = e \text{ and } i \in T_0(a)\\
        1 & c = d \text{ and } i \in T_+(a)\\
        0 & \text{otherwise}
    \end{cases}
    \enspace
    \text{and}
    \enspace
        \dsVectorUpdate{a}_{(c, i)}
        \defeq
    \begin{cases}
        1 & c = e \text{ and } i \in T_+(a)\\
        \dsVarHeight_i(a) & c = d \text{ and } i \in T(a) \\
        0 & \text{otherwise}
    \end{cases}
\]
Furthermore, the properties of $\dsOpInit$ and $\dsOpTreeDecompose$ immediately implies that each vertex $a$ can appear in at most $O(\log n)$ trees where it is not the root, so $|T(a)| = O(\log n)$. As a consequence, $\dsVectorQuery{a}$ and $\dsVectorUpdate{a}$ each has at most $O(\log n)$ non-zero entries. It is also not hard to see that we can pre-compute these vectors in $O(n \log n)$ time using our recursive $\dsOpInit$.

Applying these insights to the definition of a cycle update we see that this data structure allows each cycle update of $\simpleSolver$ to be implemented as a single dot product and a single vector addition where one vector in each operation has $O(\log n)$ non-zero entries. Since these vectors can be precomputed each cycle update can be implemented by iterating over a fixed array of $O(\log n)$ pointers and performing simple arithmetic operations. We hope this may be fast in practice.

\section{Simple Algorithm Running-Time Analysis}
\label{sec:simple_runtime}

Here we give the remaining analysis needed to prove the correctness of $\simpleSolver$ and prove \theoremref{thm:simple_algorithm}.
First we bound the energy of $\flowInitial$.

\begin{lemma}[Initial Energy]
\label{lemma:initial_energy}
\footnote{This lemma follows immediately from the well known
fact that $\tree$ is a $\stretchTotal{T}$-spectral sparsifier of $G$, cf.
\cite[Lemma 9.2]{ST08c}, but a direct proof is included here for completeness.}
$
\renergy{\flow_0} \leq \stretchTotal{\tree} \cdot
\renergy{\optFlow} \enspace.
$
\end{lemma}

\begin{proof}
Recall that $\flow_0 \in \redgevec$ is the \emph{unique} vector that meets the demands and is nonzero only on $\tree$. Now, if for every $e \in E$ we sent $\optFlow(e)$ units of flow along $\treePath{e}$ we achieve a vector with the same properties. Therefore, $\flow_0 = \sum_{e \in E} \optFlow(e) \treePathVec{e}$ and by applying the Cauchy-Schwarz inequality we get
\[
    \renergy{\flow_0}
        = \sum_{e \in \tree} r_e
        \bigg(
            \sum_{e' \in E | e \in \treePath{e'}}
            \optFlow(e')
        \bigg)^2
\leq \sum_{e \in T} \Bigg[
        \bigg(
            \sum_{e' \in E | e \in \treePath{e'}} \frac{r_e}{r_{e'}}
        \bigg)
        \bigg(
           \sum_{e' \in E | e \in \treePath{e'}}
            r_{e'} \optFlow(e')^2
        \bigg) \Bigg]
\enspace.
\]
Applying the crude bound that
$
\sum_{e' \in E | e \in \treePath{e'}}
            r_{e'} \optFlow(e')^2 \leq \renergy{\optFlow}
$
for all edges $e \in E$ and noting that by the definition of stretch
$
    \sum_{e \in T} \sum_{e' \in E | e \in T_{e'}} \frac{r_e}{r_{e'}}
    = \sum_{e' \in E} \sum_{e \in T_e} \frac{r_e}{r_{e'}}
    = \stretchTotal{\tree}
$
the result follows immediately.
\end{proof}

Next, we show how approximate optimality is preserved within polynomial factors when rounding from an $\epsilon$-approximate electric flow to tree induced voltages.

\begin{lemma}[Tree Voltage Rounding]
\label{lemma:dual_round}
Let $\flow \in \redgevec$ be a primal feasible flow with $\renergy{\flow} \leq (1 + \epsilon) \cdot \renergy{\optFlow}$ for $\epsilon > 0$ and let $\volt \in \rvertvec$ be the tree induced voltages. Then the following holds
\[
     \big\|\volt - \pseudo{\lap} \boundary \big\|_\lap
        \leq
    \sqrt{\epsilon \cdot \tau} \big\|\pseudo{\lap} \boundary\big\|_\lap \enspace.
\]
\end{lemma}

\begin{proof}
By \lemmaref{lemma:expected_progress} we know that one random cycle update from $\simpleSolver$ is expected to decrease $\renergy{\flow}$ by $\gap(\flow, \volt) / \treeCondition$. Therefore by the optimality of $\optFlow$ we have
\[
\renergy{\flow} - \frac{\gap(\flow, \volt)}{\treeCondition} \geq \renergy{\optFlow}
\]
and since $\renergy{\flow} \leq (1 + \epsilon) \cdot \renergy{\optFlow}$ and $\renergy{\flow} \leq \renergy{\optFlow}$ we have
$
    \gap(\optFlow, \volt) \leq \epsilon \cdot \tau \cdot \renergy{\optFlow}
$.
Finally, using that $\optVolt = \pseudo{\lap} \boundary$, $\optFlow = \rMatrix^{-1} \incMatrix \optVolt$, and $\lap = \incMatrix^T \rMatrix^{-1} \incMatrix$ it is straightforward to check
\[
\gap(\optFlow, \volt) = \big\|\volt - \pseudo{\lap} \boundary\big\|_\lap^2
\enspace \text{ and } \enspace
\renergy{\optFlow} = \big\|\pseudo{\lap} \boundary\big\|_\lap^2 \enspace.
\]
\end{proof}

Finally, we have everything we need to prove the correctness of $\simpleSolver$.

\begin{proof}[Proof of \theoremref{thm:simple_algorithm}]
Applying \theoremref{thm:convergence}, the definition of $\nIter$, and \lemmaref{lemma:initial_energy} we get
\[
\E\big[\renergy{\flow_\nIter}\big] - \renergy{\optFlow}
    \leq \left(1 - \frac{1}{\treeCondition}\right)^{\treeCondition \log (\epsilon^{-1} \stretchTotal{\tree} \treeCondition)} \left(\stretchTotal{\tree} \renergy{\optFlow} - \renergy{\optFlow}\right)
\enspace.
\]
Using several crude bounds and applying \lemmaref{lemma:dual_round} we get
\[
    \E\big[\renergy{\flow_\nIter}\big] \leq \left(1 + \frac{\epsilon}{\treeCondition}\right) \renergy{\optFlow}
    \enspace \text{ and } \enspace
    \E\big\|\volt_{\nIter} - \lapPseudo \boundary\big\|_{\lap}
    \leq \sqrt{\epsilon} \cdot \big\|\lapPseudo \boundary\big\|_{\lap}
\]
which is stronger than what we need to prove for the theorem.

All that remains is to bound the running time. To do this, we implement $\simpleSolver$ using the latest low-stretch spanning tree construction (\theoremref{thm:low-stretch}) and construct a spanning tree $\tree$ of stretch $\st(T)=O(m \log n \log \log n)$ in time $O(m \log n \log \log n)$ yielding $\treeCondition = O(m \log n \log \log n)$.

Next, to compute $\flow_0$ we note that given any demand vector $\boundary \in \rvertvec$ with $\sum_{i} \boundary_i = 0$, the quantity $\flow_0(e)$ on a tree edge $e\in T$ is uniquely determined by the summation of $\boundary_v$ where $v$ is over the vertices on one side of $e$. Therefore, $\flow_0$ can be computed via a DFS in $O(n)$ time.

To compute $\cycleResistance{e}$ for each off-tree edge $e \in \offtreeEdgeSet$ we could either use Tarjan's off-line LCA algorithm~\cite{Tarjan79} that runs in a total of $O(m)$ time, or simply use our own data structure in $O(m \log n)$ time. In fact, one can initiate a different instance of our data structure on $\tree$, and for each off-tree edge $(a, b) \in \offtreeEdgeSet$, one can call $\dsOpUpdate(b, 1)$ and $\dsOpUpdate(a, -1)$, so that $\cycleResistance{e} = \dsOpQuery(b) - \dsOpQuery(a) + r_e$.

Finally, we can initialize our data structure in $O(n\log n)$ time, set the initial flow values in $O(n \log n)$ time, perform each cycle update in $O(\log n)$ time, and compute all the tree voltages in $O(n\log n)$ time using the work in section \sectionref{sec:datastructure}. Since the total number of iterations of our algorithm is easily seen to be $O(m \log n \log
(n \epsilon^{-1}) \log \log n)$ we get the desired total running time.
\end{proof}

\section{Improved Running Time}
\label{sec:improvement}

In this section we show how to improve our algorithm's running time to $\runtimeBest$ (i.e. improve our algorithm's dependence on $\epsilon$ from $O(\log (\epsilon^{-1} n))$ to $O(\log \epsilon^{-1})$). The improvement is gained by careful application of two techniques to the standard $\simpleSolver$: (1) by changing the resistances of the edges in the $\tree$ and applying $\simpleSolver$, we can quickly compute an initial feasible $\flow_0 \in \redgevec$ of better quality, and (2) by changing the stopping criterion for cycle updates, we can achieve a stronger guarantee on the quality of flow induced tree voltages.

In \sectionref{sec:first-improvement}, we introduce these concepts by presenting a straightforward two-step algorithm that  improves the error dependence to $\log (\epsilon^{-1} \log n)$. In \sectionref{sec:second-improvement}, we present the full algorithm that by carefully applying these techniques recursively $O(\log^*n)$ times obtains an improved error dependence of $O(\log \epsilon^{-1})$ and the best asymptotic running time proven in this paper.

\subsection{An Example Algorithm with \texorpdfstring{$\runtimeBetter$ Running Time}{a Better Epsilon Dependence}}
\label{sec:first-improvement}

\begin{algorithm}
\SetAlgoLined
\LinesNumbered
\SetKwInOut{Input}{Input}
\SetKwInOut{Output}{Output}
\Input{$G = (V, E, r)$, $\boundary \in \rvertvec$, $\epsilon \in \rPos$}
\Output{$\flow$ and $\volt$}
\BlankLine
$T := $ low-stretch spanning tree of $G$\;
\SetKwComment{Comment}{}{}
\SetCommentSty{textrm}
$(\flow_0, \star) := \simpleSolver(G - T + \frac{T}{\log n}, \boundary, 1)$ \label{prog:warm-start-1}\;
\Comment*[f]{\emph{(the simple solver will use $T':=\frac{T}{\log n}$ as the spanning tree)}}

$\edgeSampleProb{e} := \frac{1}{\treeCondition(\tree)} \cdot \frac{\cycleResistance{e}}{r_e}$ for all
$e \in \offtreeEdgeSet$ \;
$\nIter := \big\lceil \treeCondition(T) \log \big(\frac{2\log n}{\epsilon}\big) \big\rceil$\;
$\nIter' := $ an integer chosen uniformly at random in $\{0, 1, \ldots, \lceil \tau(T) \rceil - 1\}$\;
\For{$i = 1$ to $\nIter+\nIter'$}
{
    Pick random $e_i \in \offtreeEdgeSet$ by probability distribution
    $\sampleProbVec$ \;
    $\flow_i = \flow_{i - 1} - \frac{\cyclePotential{e}{\flow_{i - 1}}}
        {\cycleResistance{e}}
    \treeCycleVec{e}$ \;
}
\Return{$\flow_{\nIter+\nIter'}$ and its tree induced voltages $\volt_{\nIter+\nIter'}$}
\caption{\label{algm:betteralgorithm}\betterSolver}
\end{algorithm}

Here we present an $\betterSolver$ (see \algorithmref{algm:betteralgorithm}), an algorithm that illustrates the techniques we will use in \sectionref{sec:second-improvement} to obtain the fastest running times presented in this paper. This algorithm modifies the resistances of tree edges and applies $\simpleSolver$ once to compute feasible $\flow_0$ that is $O(\log n)$ factor from optimal (rather than an $\tilde{O}(m)$ factor as in \lemmaref{lemma:initial_energy}) and then uses this $\flow_0$ as a warm start to perform cycle updates. By also randomizing over the number of cycle updates to perform, we can show that we get an improved running time as stated in the following theorem.

\begin{theorem}[$\betterSolver$]
\label{thm:better-algorithm}
The output of $\betterSolver$ satisfies
\[
\E\big[\renergy{\flow}\big] \leq (1 + \epsilon) \cdot \renergy{\optFlow}
\enspace \text{ and } \enspace
\E\big\|\volt - \lapPseudo \boundary\big\|_{\lap}
    \leq \sqrt{\epsilon} \cdot \big\|\lapPseudo \boundary\big\|_{\lap}
\]
and $\betterSolver$ can be implemented to run in time $\runtimeBetter$.
\end{theorem}
\noindent We prove the above theorem in two steps, introducing one new technique for each step.

\paragraph{Step \#1: Tree Scaling.}
Given graph $G$, recall that by \theoremref{thm:low-stretch} one can obtain a spanning tree $\tree$ with condition number $\treeCondition(\tree)=O(m\log n \log\log n)$. For $\kappa \geq 1 $ let $\tree / \kappa$ denote $\tree$ where every $e \in \tree$ has its resistance decreased by a multiplicative $\kappa$. Now replace $\tree$ with $\tree / \kappa$ and denote this new graph $G'$, the new resistances by $r'$, and the new tree by $\tree'$. This technique is similar to one used in \cite{KoutisMP10,KMP11}, and one can easily check that
\[
    \treeCondition(\tree') = \treeCondition(\tree)/\kappa + (m-n+1) < \treeCondition(T)/\kappa + m \enspace,
\]
i.e., decreasing the resistances of $\tree$ improves the tree condition number. Furthermore, since replacing $\tree$ by $\tree / \kappa$ only changes resistances, $\flow \in \redgevec$ is feasible in $G$ if and only if it is feasible in $G'$, and therefore one could compute an approximate flow in $G'$ and then use the results in $G$. The following lemma bounds the loss in approximation of such an operation.

\begin{lemma}[Tree Scaling]
\label{lemma:tree-scaling}
Let $G'$ be graph $G$ with $\tree$ replaced by $\tree / \kappa$ for $\kappa \geq 1$, let $\flow'$ be an $\epsilon$-approximate electrical flow in $G'$, and let $\optFlow$ be the optimal electrical flow in $G$. Then,
\[
    \renergy{\flow'} \leq \kappa (1 + \epsilon) \renergy{\optFlow} \enspace.
\]
\end{lemma}
\begin{proof}
Since $\flow'$ is feasible in $G'$ it is feasible in $G$. Letting $\optFlow'$ be the optimal electrical flow in $G'$ we can then bound $\renergy{\flow'}$ as follows
\begin{align*}
\renergy{\flow'}
&\leq \kappa \cdot \xi_{r'}(\flow') \tag{Since $r_e \leq \kappa \cdot r'_e$ for any edge $e$} \\
&\leq \kappa(1 + \epsilon) \cdot \xi_{r'}(\optFlow') \tag{Since $\flow'$ is $\epsilon$-approximate in $G'$} \\
&\leq \kappa(1 + \epsilon) \cdot \xi_{r'}(\optFlow) \tag{By optimality of $\optFlow'$ in $G'$} \\
&\leq \kappa(1 + \epsilon) \cdot \renergy{\optFlow} \enspace. \tag{Since $r_e' \leq r_e$ for any edge $e$}
\end{align*}
\end{proof}

We can apply \lemmaref{lemma:tree-scaling} directly to the analysis of $\betterSolver$. Choosing $\kappa = \log n$ in line~\ref{prog:warm-start-1}, we see that $\betterSolver$ first computes feasible $\flow_0 \in \redgevec$ that is $1$-approximate electric flow in $G'$. By having the recursive $\simpleSolver$ invocation in this line use spanning tree $\tree'$ (which has $\treeCondition(T') \leq \treeCondition(T)/\kappa+m$), this step takes $O\left( \Big(\frac{\treeCondition(T)}{\log n}+m \Big)\log n \log n \right) = O(m \log^2n \log\log n)$ time and by
\lemmaref{lemma:tree-scaling} guarantees that
\[
    \E\left[\renergy{\flow_0}\right] = (1 + \log n) \cdot \renergy{\optFlow} \leq 2 \log n \cdot \renergy{\optFlow} \enspace.
\]
This is an improvement compared to what was previously guaranteed by \lemmaref{lemma:initial_energy} and lets us apply~\theoremref{thm:convergence} with $i=\nIter=\big\lceil\treeCondition(T)\log \big(\frac{2\log n}{\epsilon} \big)\big\rceil$ to conclude that:
\begin{equation}\label{eqn:better-solver-1}
\E\big[\renergy{\flow_\nIter}\big] \leq (1+\epsilon) \renergy{\optFlow} \enspace.
\end{equation}

\paragraph{Step \#2: Randomized Stopping Time.} While \equationref{eqn:better-solver-1} provides an (expected) guarantee on the primal energy for $\flow_\nIter$, converting this to a dual guarantee on its tree induced voltage vector $\volt_\nIter$ using~\lemmaref{lemma:dual_round} causes us to lose a factor of $\treeCondition=\tilde{O}(m)$ on the dependency of $\epsilon$, which translates into an extra factor of $O(\log m)$ in the running time. This motivates us to study a more efficient conversion as follows.
\begin{lemma}
\label{lemma:randomized-stopping}
For feasible $\flow_\nIter \in \redgevec$ satisfying $\renergy{\flow_\nIter} \leq (1+\epsilon) \renergy{\optFlow}$, if one applies $\nIter'$ cycle updates, where $\nIter'$ is an integer chosen uniformly at random from $\{0,1,\dots,\lceil \treeCondition\rceil -1\}$, then we have
$$\E\big[\gap(\flow_{\nIter+\nIter'},\volt_{\nIter+\nIter'})\big] \leq \epsilon \cdot \renergy{\optFlow}  \enspace.$$
\end{lemma}
\begin{proof}
Following the notation of \lemmaref{lemma:convergence_rate}, let random variable $D_i\defeq \renergy{\flow_i} - \renergy{\optFlow}$ for $i \in \{\nIter,\nIter+1,\dots,\nIter+\lceil\treeCondition\rceil \}$. By \lemmaref{lemma:expected_progress} we have
$$
\forall i = \nIter+1,\dots,\nIter+\lceil\treeCondition\rceil, \quad \E[D_i - D_{i-1}] = -\frac{1}{\treeCondition}\E[\gap(\flow_{i-1},\volt_{i-1})]
\enspace.
$$
Now, by assumption $D_\nIter \leq \epsilon \cdot \renergy{\optFlow}$ and since $D_i \geq 0$ we can also sum over the $K'$ yielding
$$
\epsilon \cdot \renergy{\optFlow} \geq \E[D_\nIter] \geq \E[D_\nIter - D_{\nIter+\lceil\treeCondition\rceil}] \geq \frac{1}{\treeCondition}\E\left[ \sum_{\nIter'=\nIter}^{\nIter+\lceil \treeCondition \rceil-1}\gap(\flow_{\nIter'},\volt_{\nIter'}) \right] \enspace.
$$
Therefore, if one picks $\nIter' \in \{0,1,\dots,\lceil \treeCondition \rceil -1 \}$ at random, the desired inequality is obtained.
\end{proof}

Now in $\betterSolver$, our output is $\flow = \flow_{\nIter+\nIter'}$ and $\volt = \volt_{\nIter+\nIter'}$ so combining~\equationref{eqn:better-solver-1} and~\lemmaref{lemma:randomized-stopping}, we immediately have that $\E\big[\gap(\flow, \volt)\big] \leq \epsilon \cdot \renergy{\optFlow}$. This simultaneously implies the following
\begin{align*}
\E[\renergy{\flow} - \renergy{\optFlow}] = \E\big[\gap(\flow, \optVolt)\big] &\leq \E\big[\gap(\flow, \volt)\big] \leq \epsilon \cdot \renergy{\optFlow} \enspace, \text{ and} \\
\E\big[\|\volt - \pseudo{\lap} \boundary\|_\lap\big]^2 \leq \E\big[\|\volt - \pseudo{\lap} \boundary\|_\lap^2\big] = \E\big[\gap(\optFlow, \volt)\big] &\leq \E\big[\gap(\flow,\volt)\big] \leq \epsilon \cdot \renergy{\optFlow} = \epsilon \|\pseudo{\lap} \boundary\|_\lap^2 \enspace.
\end{align*}
Here the second line has used the fact that $\gap(\optFlow, \volt) = \|\volt - \pseudo{\lap} \boundary\|_\lap^2$ and $\renergy{\optFlow}=\|\pseudo{\lap} \boundary\|_\lap^2$ whose proof can be found in the proof of~\lemmaref{lemma:dual_round}. This finishes the proof of~\theoremref{thm:better-algorithm}.

\subsection{The Full Algorithm With \texorpdfstring{$\runtimeFinal$ Running Time}{our Best Epsilon Dependence}}
\label{sec:second-improvement}
Here we present $\finalSolver$ (see \algorithmref{algm:finalalgorithm}), an algorithm that achieves the best asymptotic running time presented this paper by recursively applying the tree scaling technique of the previous section $\log^*(n)$ number of times with careful chosen scaling parameters.

\begin{theorem}[$\finalSolver$]
\label{thm:final-algorithm}
The output of $\finalSolver$ satisfies
\[
\E\big[\renergy{\flow}\big] \leq (1 + \epsilon) \cdot \renergy{\optFlow}
\enspace \text{ and } \enspace
\E\big\|\volt - \lapPseudo \boundary\big\|_{\lap}
    \leq \sqrt{\epsilon} \cdot \big\|\lapPseudo \boundary\big\|_{\lap}
\]
and $\finalSolver$ can be implemented to run in time $\runtimeFinal$.
\end{theorem}
\begin{algorithm}
\SetAlgoLined
\LinesNumbered
\SetKwInOut{Input}{Input}
\SetKwInOut{Output}{Output}
\Input{$G = (V, E, r)$, $\boundary \in \rvertvec$, $\epsilon \in \rPos$}
\Output{$\flow$ and $\volt$}
\BlankLine
$T := $ low-stretch spanning tree of $G$\;
$\flow^{(0)} := $ unique flow on $T$ such that $\incMatrix^T \flow_0 = \boundary$\;
Define $\kappa^{(i)}$ and $\nIter^{(i)}$ according to \equationref{eqn:final-solver-kappa} and \equationref{eqn:final-solver-niter}\;
\For{$i = 1$ to $c+1$}
{
$\flow_0 := \flow^{(i-1)}$\;
$G^{(i)} := G - \tree + \tree^{(i)}$ where $\tree^{(i)}:=\frac{\tree}{\prod_{j=i}^{c} \kappa_j}$\;
\SetKwComment{Comment}{}{}
\SetCommentSty{textrm}
\Comment*[f]{\emph{(the symbols $R_e, r_e, \cyclePotential{e}{\cdot}$ below are with respect to this new graph $G^{(i)}$ and tree $\tree^{(i)}$)}}

$\edgeSampleProb{e} := \frac{1}{\treeCondition(\tree^{(i)})} \cdot \frac{\cycleResistance{e}}{r_e}$ for all
$e \in \offtreeEdgeSet$ \;
\For{$j = 1$ to $\nIter^{(i)}$}
{
    Pick random $e_i \in \offtreeEdgeSet$ by probability distribution
    $\sampleProbVec$ \;
    $\flow_i = \flow_{i - 1} - \frac{\cyclePotential{e}{\flow_{i - 1}}}
        {\cycleResistance{e}}
    \treeCycleVec{e}$ \;
}
$\flow^{(i)}:=\flow_{\nIter^{(i)}}$\;
}
\Return{$\flow^{(c+1)}$ and its tree induced voltages $\volt^{(c+1)}$}
\caption{\label{algm:finalalgorithm}\finalSolver}
\end{algorithm}
\noindent
We begin our proof of the above theorem by first making the parameters in $\finalSolver$ explicit. Again, $T$ is a low-stretch spanning tree of $G$ with $\treeCondition(T)=O(m\log n \log\log n)$. Let us define a sequence $(\kappa_1,\dots,\kappa_c)$ where
\begin{equation}\label{eqn:final-solver-kappa}
\kappa_1 = \log n, \quad \kappa_i = \log \kappa_{i-1}, \quad \text{and $c$ is chosen so that $\kappa_c \leq 2$.}
\end{equation}
For each $i\in[c+1]$, we define $T^{(i)}$ to be the same spanning tree $T$ but with resistances on each edge decreased by a factor of $\prod_{j=i}^{c} \kappa_j$, and let $G^{(i)}$ be the graph $G$ with its spanning tree $T$ replaced by $T^{(i)}$. Notice that $T^{(c+1)}=T$ and $G^{(c+1)}=G$. By a similar observation as in~\sectionref{sec:first-improvement}, we have that $\treeCondition(T^{(i)}) \leq \frac{\treeCondition(T)}{\prod_{j=i}^c \kappa_j} + m$ due to the scaling of the tree edges.

We now choose the number of iterations
\begin{equation}
\label{eqn:final-solver-niter}
\nIter^{(i)} =
\left\{
  \begin{array}{ll}
    \Big\lceil \big(\frac{\treeCondition(T)}{\prod_{j=1}^{c} \kappa_j} + m \big) \log(\st(T^{(1)})) \Big\rceil, & \hbox{$i=1$;} \\[1em]
    \Big\lceil \big(\frac{\treeCondition(T)}{\prod_{j=i}^{c} \kappa_j} + m \big) \log(2\kappa_{i-1}-1) \Big\rceil, & \hbox{$i=2,3,\dots,c$;} \\[1em]
    \Big\lceil \big(\treeCondition(T) + m \big) \log(\frac{2\kappa_{c}-1}{\epsilon}) \Big\rceil + \nIter', & \hbox{$i=c+1$.}
  \end{array}
\right.
\end{equation}
where $\nIter'$ is chosen from $\{0,1,\dots,\lceil \treeCondition \rceil -1 \}$ uniformly at random just like that in~\sectionref{sec:first-improvement}. Those choices induce the following property on flow vectors $\flow^{(i)}$ constructed in $\finalSolver$:
\begin{claim}
For each $i\in [c]$, $\flow^{(i)}$ is (in expectation) a $1$-approximate electrical flow for $G^{(i)}$ with demand vector $\boundary$, and also (in expectation) a $(2\kappa_i-1)$-approximate electrical flow for $G^{(i+1)}$.
\end{claim}
\begin{proof}
Let $r^{(i)} \in \redgevec$ denote the resistance vector in graph $G^{(i)}$, and let $\optFlow^{(i)} \in \redgevec$ denote the electrical flow satisfying $\boundary \in \rvertvec$ in graph $G^{(i)}$. We prove the claim by induction on $i$.

In the base case ($i=1$), since $\flow^{(0)}$ is the unique flow on the tree meeting the demands \lemmaref{lemma:initial_energy} implies that $\xi_{r^{(1)}}\big( \flow^{(0)} \big) \leq \st(T^{(1)}) \xi_{r^{(1)}}\big(\optFlow^{(1)}\big)$. By our choice of $\nIter^{(1)}\geq \treeCondition(T^{(1)})\log(\st(T^{(1)}))$ and~\theoremref{thm:convergence}, we then see that $\E\big[\xi_{r^{(1)}}\big( \flow^{(1)} \big) \big] \leq (1+1) \xi_{r^{(1)}}\big(\optFlow^{(1)}\big)$ and therefore the computed $\flow^{(1)}$ is in expectation $1$-approximate for $G^{(1)}$. Furthermore, since the resistances in $T^{(1)}$ and $T^{(2)}$ are off by a multiplicative $\kappa_1 = \log n$, \lemmaref{lemma:tree-scaling} implies that $f^{(1)}$ satisfies $\E\big[\xi_{r^{(2)}}\big( \flow^{(1)} \big) \big] \leq 2\kappa_1 \xi_{r^{(2)}}\big(\optFlow^{(2)}\big)$ so is in expectation $(2\kappa_1-1)$-approximate for $G^{(2)}$.

For $i\in\{2,3,\dots,c\}$, we have by our inductive hypothesis that $\flow^{(i-1)}$ is in expectation $(2\kappa_{i-1}-1)$-approximate for $G^{(i)}$. Therefore, by our choice of $\nIter^{(i)}\geq \treeCondition(\tree^{(i)}) \log(2\kappa_{i-1}-1)$ and \theoremref{thm:convergence}, we immediately see that the constructed $\flow^{(i)}$ satisfies $\E\big[\xi_{r^{(i)}}\big( \flow^{(i)} \big) \big] \leq (1+1) \xi_{r^{(i)}}\big(\optFlow^{(i)}\big)$ and is in expectation a $1$-approximate for $G^{(i)}$. Furthermore, since the resistances in $T^{(i)}$ and $T^{(i+1)}$ are off by a multiplicative $\kappa_i$, \lemmaref{lemma:tree-scaling} implies that $f^{(i)}$ satisfies $\E\big[\xi_{r^{(i+1)}}\big( \flow^{(i)} \big) \big] \leq 2\kappa_i \xi_{r^{(i+1)}}\big(\optFlow^{(i+1)}\big)$ so is in expectation $(2\kappa_i-1)$-approximate for $G^{(2)}$.
\end{proof}

As a direct consequence of the above lemma we see that in an execution of $\finalSolver$, $\flow^{(c)}$ is in expectation a $(2\kappa_c-1)$-approximate electrical flow in graph $G=G^{(c+1)}$. Therefore, when $\flow^{(c)}$ is used as an initial, $\flow_0$, for the last iteration ($i = c + 1$), after $\nIter:=\nIter^{(c+1)}-\nIter'\geq \treeCondition(\tree) \log(\frac{2\kappa_{c}-1}{\epsilon})$ cycle updates we have
$
\E\big[\renergy{\flow_\nIter}\big] \leq (1+\epsilon) \renergy{\optFlow}
$. Therefore, by the same analysis of randomized stopping time as in~\sectionref{sec:first-improvement}, after running the final $\nIter'$ cycle updates the final flow $\flow^{(c + 1)} =\flow_{\nIter+\nIter'}$ and $\volt^{(c + 1)} = \volt_{\nIter+\nIter'}$ satisfy the error guarantee of~\theoremref{thm:final-algorithm}.
\[
\E[\renergy{\flow^{(c + 1)}}] \leq (1 + \epsilon) \cdot \renergy{\optFlow}
\enspace \text{ and } \enspace
\E\|\volt^{(c + 1)} - \lapPseudo \boundary\|_{\lap}
    \leq \sqrt{\epsilon} \cdot \|\lapPseudo \boundary\|_{\lap}
\]

To prove \theoremref{thm:final-algorithm} all that remains is to bound the running time of $\finalSolver$. As in $\simpleSolver$ each cycle update takes $O(\log n)$ time so the total running time of $\finalSolver$ is

\begin{align*}
O\left( \sum_{i=1}^{c+1} \nIter^{(i)} \log n \right)
& = \runtimeFinal + O\left( \sum_{i=2}^{c} \nIter^{(i)} \log n \right)
\enspace.
\end{align*}
To upper bound the second term we note the following
\begin{align*}
O\left( \sum_{i=2}^{c} \nIter^{(i)} \log n \right)
& = O\left( \sum_{i=2}^{c}  \Big(\frac{\treeCondition(T)}{\prod_{j=i}^{c} \kappa_j} + m \Big) \log(2\kappa_{i-1}-1) \log n \right) \\
& = O\left( \sum_{i=2}^{c}  \Big(\frac{\treeCondition(T)}{\prod_{j=i}^{c} \kappa_j} + m \Big) \kappa_i \log n \right)  \tag{using $\kappa_i = \log \kappa_{i-1}$}\\
& = O\left( \sum_{i=2}^{c}  \frac{\treeCondition(T)}{\prod_{j=i+1}^{c} \kappa_j} \log n +
            \sum_{i=2}^{c}  m \kappa_i\log n \right) \\
& = O\left( \treeCondition(T) \log n + m \kappa_2\log n \right) = O(m \log^2 n \log \log n) \enspace.
\end{align*}
The last step used that all $\kappa_j \geq 2$ so $\sum_{i=2}^{c} \frac{1}{\prod_{j=i+1}^{c} \kappa_j} \leq 1 + \frac{1}{2} + \frac{1}{4} + \cdots = O(1)$ and $\sum_{i\geq 2} \kappa_i = O(\kappa_2)$. This ends the proof of~\theoremref{thm:final-algorithm} with the desired running time $\runtimeFinal$.

\section{Numerical Stability}
\label{sec:stability}

Up until this point our analysis has assumed that all arithmetic operations are exact. In this section, we show that our algorithm is numerically stable and achieves the same convergence guarantees when implemented with finite-precision arithmetic. We start with the analysis of $\simpleSolver$.

For simplicity of exposition, we assume that the resistances $r_e$ and the coordinates of the demand vector $\boundary$ are all represented as $b$-bit integers with absolute values bounded by $N=2^b$. We show that $\simpleSolver$ works when arithmetic operations are implemented with $O(\max (b,\log n, \log 1/\epsilon))$ bits of precision (which is necessary simply to guarantee we can represent an answer meeting the required error bound).  In particular, if the entries of the input and $\epsilon$ can be stored in $\log n$-bit words,  our running time guarantees hold in the standard unit cost RAM model, which only allows arithmetic operations on $O(\log n)$-bit numbers to be done  in constant time.

We start our algorithm by multiplying the demand vector $\boundary$ by $\lceil 4mN^2/\epsilon  \rceil$, to ensure that in $\optFlow$ there exist at least $\lceil 4mN^2/\epsilon  \rceil$ total units of flow on the edges. Since the resistances are at least 1, this guarantees that $\renergy{\optFlow} \geq \lceil 4mN^2/\epsilon  \rceil$. Next, we ensure that throughout our algorithm all flow vectors $\flow_i$ are integer vectors. At each iteration of $\simpleSolver$, when we are given a primal feasible flow $\flow$ and want to compute the optimal cycle update $\alpha^* \defeq \frac{\flow^T \rMatrix \treeCycleVec{e}}{\treeCycleVec{e}^T \rMatrix \treeCycleVec{e}}$, we round this fraction to the nearest integer and suppose we pick some $\tilde{\alpha}\in\mathbb{Z}$ such that $|\alpha^*-\tilde{\alpha}|\leq \frac{1}{2}$.

Based on the modifications above, it is easy to see that our algorithm can be run on a RAM machine with word size $O(\max(b,\log n, \log 1/\epsilon))$,  since all flow values are integers bounded above by $O(\poly(N,n,1/\epsilon))$. All we have left to show is the convergence analysis for such integral updates.

We do so by strengthening our cycle update lemma, \lemmaref{lemma:cycle-update}. For each cycle update, if we add  $\tilde{\alpha}=\alpha^*(1+\delta)$ units of flow on the cycle instead of $\alpha^*$,  the corresponding energy decrease is
\begin{align*}
 \renergy{\flow - \alpha^*(1 + \delta)\treeCycleVec{e}} - \renergy{\flow}
   &=\ (\flow - \alpha^*(1 + \delta) \treeCycleVec{e})^T \rMatrix (\flow - \alpha^*(1 + \delta)\treeCycleVec{e}) - \flow^T \rMatrix \flow
   \\&=\ - 2 \alpha^*(1+\delta) \flow^T \rMatrix \treeCycleVec{e}
        + \left(\alpha^*(1+\delta)\right)^2 \treeCycleVec{e}^T \rMatrix \treeCycleVec{e}
   = - \frac{\cyclePotential{e}{\flow}^2}{\cycleResistance{e}}(1 - \delta^2) \enspace,
\end{align*}
where the last equality has used the fact that $\cycleResistance{e} = \treeCycleVec{e}^T \rMatrix \treeCycleVec{e}$. We have $|\delta| \leq \frac{1}{2\alpha^*}$, so, as long as $\alpha^*\geq 1$, the decrease in the energy decrease is at least $\frac{3}{4}$ of what it would be for $\delta=0$. We call an off-tree edge ``good'' if its $\alpha^*\geq 1$, and ``bad'' otherwise.  We can rewrite the duality gap as
\begin{align*}
\gap(\flow, \volt)
 &= \sum_{e \in E \setminus \tree} \frac{\cyclePotential{e}{\flow}^2}{r_e}
 = \sum_{e \in E \setminus \tree} \frac{\cyclePotential{e}{\flow}^2}{R_e} \frac{R_e}{r_e}
= \sum_{\substack{e \in E \setminus \tree\\\text{$e$ is bad}}} \frac{\cyclePotential{e}{\flow}^2}{R_e^2} \frac{R_e^2}{r_e}
 + \sum_{\substack{e \in E \setminus \tree\\\text{$e$ is good}}} \frac{\cyclePotential{e}{\flow}^2}{R_e} \frac{R_e}{r_e} \\
 &\leq mN^2 + \sum_{\substack{e \in E \setminus \tree\\\text{$e$ is good}}} \frac{\cyclePotential{e}{\flow}^2}{R_e} \frac{R_e}{r_e} \enspace.
\end{align*}

As a consequence, if one  samples\footnote{Recall that \emph{exact} sampling can be done in \emph{expected} constant time on a machine with $O(\log n)$-sized words.} each tree cycle $\treeCycleVec{e}$ with probability to $\frac{R_e}{r_e\treeCondition}$, then the expected energy decrease is at least:
\begin{align*}
 \E&\left[\renergy{\flow_{i}} - \renergy{\flow_{i-1}} \middle| \gap(\flow_{i-1},\volt_{i-1})\right]
\leq \sum_{\substack{e\in E\setminus \tree\\\text{$e$ is good}}} \left( \frac{R_e}{r_e\treeCondition} \right)\left( - \frac{\cyclePotential{e}{\flow}^2}{\cycleResistance{e}}\frac{3}{4}\right) \\
\leq& -\frac{\left(\gap(\flow_{i-1},\volt_{i-1}) - mN^2\right)}{4\treeCondition/3}
\leq -\frac{\left(\renergy{\flow_{i-1}}-\renergy{\optFlow} - mN^2\right)}{4\treeCondition/3} \enspace.
\end{align*}

If one defines a random variable $D_i\defeq \renergy{\flow_i} - \renergy{\optFlow} - mN^2$, then
$$ \E[D_i | D_{i-1}] \leq \left(1-\frac{1}{4\treeCondition/3}\right)D_{i-1} \enspace.$$
Therefore, using the same analysis as before, after $\nIter=\frac{4\treeCondition}{3} \log\frac{2\stretchTotal{\tree}}{\epsilon \cdot p}$, we have   with probability at least $1-p$ that:
\begin{align*}
\renergy{\flow_\nIter} \leq \left(1+\frac{\epsilon}{2}\right) \left(\renergy{\optFlow} + mN^2\right) \leq (1+\epsilon)\renergy{\optFlow}\enspace.
\end{align*}

We note that the above analysis extends to $\betterSolver$ and $\finalSolver$. In fact, one can always round those scaling factor $\kappa$'s to their nearest integers and verify that the same analysis for running time and convergence holds. This ensures that throughout our algorithm each graph $G^{(i)}$ has its edge resistances all being rational numbers with the same integer denominator that is bounded above by $\tilde{O}(\log n)$. Therefore, when solving the approximate electrical flow for each graph $G^{(i)}$, we can simply scale up all resistances by their common denominator, and this reduces the analysis on numerical stability to the earlier case of $\simpleSolver$.

\section{A Geometric Interpretation via Alternating Projections} \label{sec:geoview}

In this section, we present a geometric interpretation of our algorithm and provide an alternative analysis of its convergence based on the Method of Alternating Projections~\cite{BoydV, EscalanteR}.
This method provides a generic framework for iteratively finding a point in the intersection of a number of convex constraint sets by sequentially projecting the current solution onto a set at the time.
The specialization of this method to systems of linear equations is known as the Kaczmarz method~\cite{Kaczmarz} and is commonly used in practice, especially in the field of Computerized Tomography~\cite{Natterer} to solve large linear systems efficiently.

An important contribution towards a better understanding of the Kaczmarz method was given by Strohmer and Vershynin~\cite{SV}, who showed that a randomized version of such method converges exponentially with a rate equal to
the scaled condition number of the system.
We remark that $\simpleSolver$ can be interpreted as applying this randomized Kaczmarz method to the scaled system of KPL equations:
$$ \forall e \in \offtreeEdgeSet \enspace , \enspace \frac{1}{\sqrt{r_e}} \flow^T \rMatrix \circVec_{e} = 0\enspace.
$$
Our analysis effectively shows that the tree condition number $\tau$ plays exactly the role of the scaled condition number of the system in the randomized Kaczmarz method.

\subsection{Geometric View}

Given the low-stretch spanning tree $T$ and the corresponding basis $\{\circVec_e\}_{e \in \offtreeEdgeSet}$ of the cycle space of $G$, let $P_e \defeq \{\flow \in \R^E : \flow^T \rMatrix \circVec_e = 0\}$ be the hyperplane of flows $\flow$ respecting the KPL condition over circulation $\circVec_e$.
Then, the optimality condition with respect to a basis of the cycle space can be viewed as requiring the flow $\flow$ to be in the intersection of the $(m-n+1)$ hyperplanes $\{P_e\}_{e \in \offtreeEdgeSet}$.

From this geometric perspective, at every iteration, our algorithm picks a hyperplane $P_e$ associated with a basis vector $\circVec_e$, and projects the current flow $\flow_i$ onto $P_e$.
Formally, we can check that, at iteration $i,$ our algorithms applies to the current flow the projection $\Pi_{e_i}$ onto $P_{e_i}$:
\begin{equation}\label{eq:cycleproj}
\Pi_{e_i} \flow_i = \bigg(\iMatrix - \frac{\circVec_{e_i} \circVec_{e_i}^T \rMatrix} {\| \circVec_{e_i} \|^2_\rMatrix}\bigg) \flow_i = \flow_i - \frac{\flow_i^T \rMatrix \circVec_{e_i}}{\| \circVec_{e_i} \|^2_\rMatrix} \circVec_{e_i}  = \flow_{i+1}
\enspace.
\end{equation}
Notice that, as this update adds a circulation to $\flow_i$, the resulting flow $\flow_{i+1}$ meets the demands $\boundary$.

Our analysis shows that by iteratively projecting $\flow_i$ onto different hyperplanes in $\{P_e\}_{e \in \offtreeEdgeSet}$, the final flow can be made arbitrarily close to the intersection $\bigcap_{e \in \offtreeEdgeSet} P_e$, i.e. the unique electrical flow.

\subsection{Alternating Projections}

The geometric interpretation above casts our algorithm as an instance of the Method of Alternating Projections~\cite{BoydV, EscalanteR}. This method is an algorithmic paradigm used to solve the following generic problem. Given $k$ constraints in the form of convex sets $S_1, \ldots, S_k \in \R^n$ such that $\bigcap_{i=1}^k S_i$ is non-empty, and we have the ability to project onto each individual $S_i,$ i.e. given $x \in \R^n$ we can efficiently compute the projected point $\Pi_i(x) \in S_i$ that is closest to $x$. The goal is to find (or come arbitrarily close to) a point in  $\bigcap_{i=1}^k S_i$.

A generic alternating-projection algorithm solves this problem by iteratively applying a sequence of projections $\Pi_{i_1}, \Pi_{i_2}, \ldots$ to a starting point $x_0 \in \R^n$. Given the current solution $x_t,$ the next iterate is defined as $x_{t+1} = \Pi_{i_t}(x_t)$.
In the case of our electrical-flow algorithm, the convex sets $\{S_i\}$ consist of the hyperplanes $\{P_e\}_{e \in E\setminus T}$ and the notion of distance is given by the resistance norm  $\|x\|_\rMatrix = \sqrt{x^T \rMatrix x}$.

The analysis of the convergence of alternating-projection algorithms exploits the convexity of the sets to argue that, at each iteration $t$, for any $x^*$ in the target set $\bigcap_{i=1}^n S_i,$
\begin{equation}
\|x_{t+1} - x^*\|^2 \leq \|x_{t} - x^* \|^2 - \|x_{t+1} - x_{t}\|^2 \enspace. \label{eqn:altproj}
\end{equation}
In words, the additive progress in the direction of $x^*$ in one iteration is at least the squared distance between $x_{t+1}$ and $x_{t},$ i.e. the amount our solution has moved as a result of the projection.
Hence, the convergence is fastest when we are able to pick the next projection to be onto a constraint that is substantially violated.

\subsection{A Randomized Kaczmarz Algorithm}
When the convex sets are hyperplanes, the alternating-projection algorithm specializes to the Kaczmarz method~\cite{Kaczmarz}. Different variants of this method have proved effective at solving systems of sparse linear equations in computerized tomography~\cite{Natterer} and image reconstruction~\cite{Herman}.

The convergence analysis still relies on \equationref{eqn:altproj}. For all $e \in \offtreeEdgeSet,$ define the normalized vector $\hat{c}_e \defeq \frac{\circVec_e}{\| \circVec_e\|_\rMatrix}$.
For the specific case of our algorithm, \equationref{eqn:altproj}  yields
\begin{equation} \label{eqn:potential}
\|\flow_{t} - \optFlow\|_\rMatrix^2 - \|\flow_{t+1} - \optFlow\|_\rMatrix^2 \geq \|\flow_{t+1} - \flow_t\|_\rMatrix^2 = (\flow_t^T \rMatrix \hat{c}_{e_t})^2 \enspace.
\end{equation}
Equivalently, we make fast progress at iteration $t$ if we project over $P_{e_t}$ such that $(\flow_t^T \rMatrix \hat{c}_{e_t})^2$ is large with respect to our current distance from optimal.

Recalling that, by \lemmaref{lemma:kpl} at optimality $\optFlow^T \rMatrix \circVec =0$ for all circulations $\circVec$, we can also rephrase our task as finding a basis direction $\hat{c}_{e}$, $e \in \offtreeEdgeSet$, such that $\big( (\flow_t - \optFlow)^T \rMatrix \hat{c}_{e} \big)^2$ is a large fraction of $\|\flow_{t} - \optFlow\|_\rMatrix^2,$ i.e. $\hat{c}_{e}$ captures a large component of the current error $\flow_t - \optFlow$.

This observation raises two main challenges. First, there may be no basis direction $\hat{c}_e$ that is well correlated with the error. Second, even if such a basis direction exists, the computation required to find it may be too slow for our purposes, e.g. it might require looking at $\circVec_e$ for all $e \in \offtreeEdgeSet$.

Our algorithm solves both these issues by a randomized approach, that is to sample a direction $\hat{c}_e$ for $e \in \offtreeEdgeSet$  with probability $p_e$ proportional to $\frac{R_e}{r_e}$, and then apply the corresponding projection $\Pi_e$. We prove that this makes sufficient progress in expectation by exploiting the fact that the basis $\{\circVec_e\}_{e \in E\setminus T}$ comes from a low-stretch spanning tree. In particular, the following property of the basis $\{\circVec_e\}_{e \in \offtreeEdgeSet}$ is crucial to our argument, and its proof is similar to that of \lemmaref{lemma:expected_progress}.

\begin{lemma} \label{lemma:geocondition}
For any circulation $\vec{g},$ we have
\begin{equation}
\sum_{e \in \offtreeEdgeSet} p_e (\vec{g}^T \rMatrix \hat{c}_{e})^2 \geq \frac{\|\vec{g}\|^2_\rMatrix}{\treeCondition}
\enspace.
\label{eqn:treecover}
\end{equation}
\end{lemma}
\begin{proof}
Recall that, for $e = (a,b) \in \offtreeEdgeSet$, $\treePathVec{e} \in \redgevec$ is the vector corresponding to the unique path from $a$ to $b$ by $\tree$.
By the definition of $\treeCycleVec{e},$ we have:
\begin{align*}
\sum_{e \in \offtreeEdgeSet} p_e (\vec{g}^T \rMatrix \treeCycleVec{e})^2 &=
\frac{1}{\treeCondition} \sum_{e \in \offtreeEdgeSet} \frac{\left( g_e r_e + \vec{g}^T \rMatrix (- \treePathVec{e}) \right)^2}{r_e} \\
\geq \frac{1}{\treeCondition} \left(\sum_{e \in \offtreeEdgeSet} r_e g_e^2
- 2 \sum_{e \in \offtreeEdgeSet} g_e \cdot (\vec{g}^T \rMatrix \treePathVec{e}) \right)
&= \frac{1}{\treeCondition} \left(\sum_{e \in \offtreeEdgeSet} r_e g_e^2
+ 2 \vec{g}^T \rMatrix \bigg(- \sum_{e \in \offtreeEdgeSet} g_e \treePathVec{e} \bigg)\right).
\end{align*}
Since $\vec{g}$ is a circulation, its values on the off-tree edges determines its values on the tree edges. Consequently,
$$
\forall e' \in E \enspace \text{:} \enspace
\left(- \sum_{e \in \offtreeEdgeSet} g_e \vec{p}_e\right)_{e'} =
\begin{cases}
	g_{e'} & e' \in T\\
	0		 & e' \notin T\\
\end{cases}
\enspace.
$$
This yields:
$$
\sum_{e \in \offtreeEdgeSet} p_e (\vec{g}^T \rMatrix \hat{c}_{e})^2 \geq
  \frac{1}{\tau} \left(\sum_{e \in \offtreeEdgeSet} r_e g_e^2
+ 2 \sum_{e \in T} r_{e} g_{e}^2 \right) \geq \frac{\|\vec{g}\|^2_\rMatrix}{\treeCondition}.
$$
\end{proof}

The geometric meaning of this lemma is that, in expectation under distribution $\vec{p}$, any circulation $\circVec$ has correlation at least $\frac{1}{\tau}$ with the basis directions $\{\circVec_e\}_{e \in E\setminus T}$.
We remark here that the uniform distribution over an arbitrary orthonormal basis of the cycle space would satisfy \equationref{eqn:treecover} with $\treeCondition$ replaced by the number of dimensions $(m - n + 1)$, and would yield the best possible correlation bound for our randomized approach. However, it is not known how to produce such a basis efficiently enough. Moreover, we also rely on the compact representation of our basis $\{\circVec_e\}_{e \in E\setminus T}$ as a spanning tree to design our data structure in \sectionref{sec:datastructure}.
By \lemmaref{lemma:geocondition}, we obtain that our expected progress is
\begin{align*}
\E_{e_t \leftarrow \vec{p}} \left[ \big( (\flow_t - \optFlow)^T \rMatrix \hat{c}_{e_t} \big)^2 \right] =
\sum_{e \in \offtreeEdgeSet} p_e \big((\flow_t - \optFlow)^T \rMatrix \hat{c}_{e}\big)^2 \geq
 \frac{1}{\treeCondition} \cdot \|\flow_t - \optFlow \|^2_\rMatrix\enspace.
\end{align*}
Hence, by \equationref{eqn:potential}, the expected distance of our current solution from optimal decreases as
\begin{equation}
\E \left[ \|\flow_{t} - \optFlow\|_\rMatrix^2 \right] \leq \Big(1 - \frac{1}{\treeCondition}\Big)^t \|\flow_{0} - \optFlow\|_\rMatrix^2 \enspace. \label{eqn:potred}
\end{equation}
To bound our initial distance from optimum, we we note that for any feasible $\flow \in \redgevec$ the optimality of $\optFlow$ implies
$$
\|\flow\|^2_\rMatrix = \|\optFlow\|^2_\rMatrix + \|\flow - \optFlow\|^2_\rMatrix \enspace.
$$
Then, by \lemmaref{lemma:initial_energy}, it must be the case that  $\|\flow_0 - \optFlow\|^2_\rMatrix \leq \stretchTotal{T} \cdot \|\optFlow\|^2_\rMatrix$.
Combined with \equationref{eqn:potred}, this shows that after $\nIter = \treeCondition \log \frac{\stretchTotal{\tree}}{\epsilon}$ iterations,
$ \E \left[ \|\flow_{\nIter} - \optFlow\|_\rMatrix^2 \right]\leq \epsilon \|\optFlow\|^2_\rMatrix$.

\section{An Operator View: Linearity and Approximating \texorpdfstring{$\pseudo{\lap}$}{The Laplacian Pseudoinverse}}
\label{sec:op-view}

In this section, we depart from the electrical perspective of \sectionref{sec:algorithm} and the geometric one of \sectionref{sec:geoview} to present an interpretation of the $\simpleSolver$ algorithm as a composition of simple linear operators.
The main results of this section are \theoremref{thm:linearity}, which shows the linearity of $\simpleSolver$, and \theoremref{thm:projapprox} and its corollaries, which prove that the linear operator corresponding to an execution of $\simpleSolver$ is an approximation of the Laplacian pseudoinverse with large probability.

\subsection{Linearity of \texorpdfstring{$\simpleSolver$}{simpleSolver}}

In \sectionref{sec:geoview} we showed that each cycle update of $\simpleSolver$ is the application of projector $\Pi_{e_i}$ to $\flow_i$ where $e_i$ is the off-tree edge sample in iteration $i$ and the projectors are defined by
$$
\forall e \in \offtreeEdgeSet \enspace : \enspace
\Pi_e \defeq
\bigg(\iMatrix - \frac{\treeCycleVec{e} \treeCycleVec{e}^T \rMatrix} {\| \circVec_{e} \|^2_\rMatrix}\bigg)\enspace.
$$
The projectors $\{\Pi_e\}_{e \subset \offtreeEdgeSet}$ as well as the incidence matrix for the spanning tree $T$, denoted $\incMatrix_T \in \R^{E \times V}$, and the Laplacian of $\tree$, denoted $\lap_T \defeq \incMatrix_T^T \rMatrix^{-1} \incMatrix_T$, give us the necessary building blocks to describe $\simpleSolver$ as a linear operator.

\begin{theorem}
\label{thm:linearity}
Consider an execution of $\simpleSolver$ on demand vector $\boundary$ for $\nIter$ iterations of cycle updates.
Let $e_1, \ldots, e_\nIter$ be the random off-tree edges picked from the
probability distribution $\{p_e\}$ by the algorithm.
Then, the final flow $\flow_\nIter$ and tree induced voltages $\volt_\nIter$ are given by
\[
\flow_\nIter = \mvar{F}_\nIter \boundary \enspace, \enspace \volt_\nIter = \mvar{V}_\nIter \boundary,
\]
where the \emph{flow operator} $\mvar{F}_\nIter$ and the \emph{voltage operator} $\mvar{V}_\nIter$ are defined as:
\[
\mvar{F}_\nIter \defeq \left(\prod_{j = 1}^\nIter \Pi_{e_j} \right) \rMatrix^{-1} \incMatrix_T \pseudo{\lap_T}
\enspace, \enspace
 \mvar{V}_\nIter \defeq \pseudo{\lap_T} \incMatrix_T^T  \left(\prod_{j = 1}^\nIter \Pi_{e_j} \right)  \rMatrix^{-1}
\incMatrix_T
\pseudo{\lap_T}.
\]
\end{theorem}
\begin{proof}
The initial flow $\flow_0$ is the electrical flow over the tree $T$ and therefore
$
    \flow_0 = \rMatrix^{-1} \incMatrix_T \pseudo{\lap_T} \boundary
$.
Since applying $\Pi_e$ corresponds to a cycle update we have
$
    \flow_{\nIter} = \left(\prod_{j = 1}^\nIter \Pi_{e_j} \right) f_0
$.
Furthermore, we see that the final tree induced voltages for $\flow_{\nIter}$ can be obtained by
$
\volt_{\nIter} =  \lapPseudo \incMatrix_T^T \flow_{\nIter}
$. Combining these facts yields the result.

\end{proof}

The ideas of \theoremref{thm:linearity} apply to $\finalSolver$ and one can show that for a fixed choice of sampled edges $\finalSolver$ is also linear operator. Furthermore, for $\simpleSolver$ we see that we could make the flow and voltage operators symmetric without a loss in asymptotic running time or error guarantees by simply applying the projections a second time in reverse.

\subsection{Approximation of Laplacian Pseudoinverse}
\newcommand{\Proj}{\tilde{\Pi}}
\newcommand{\norm}[1]{\left\lvert \left\lvert #1 \right\rvert \right\rvert}

In this section, we show that the flow and voltage operators $\mvar{F}_\nIter$ and $\mvar{V}_\nIter$ arising from $\simpleSolver$  yield an approximate representation of the pseudoinverse $\lap^\dagger$ and related operators as the composition of a number of simple linear operators.
Our notion of approximation will be with respect to the Frobenius norm. We denote the Frobenius norm of a matrix $\mvar{X}$ by $\norm{\mvar{X}}_F$ and recall that
$\norm{\mvar{X}}_F^2 \defeq \Tr(\mvar{X}^T \mvar{X}).$

For the purpose of this section, the natural notion of inner product in $\R^E$ is given by the resistance operator $\rMatrix.$ To make the proof simpler to follow, we scale some of the operators defined so far as to be able to use the standard notion of inner product. We define
$$
\forall e \in \offtreeEdgeSet ~ : ~
\Proj_e \defeq \rMatrix^{1/2} \Pi_e \rMatrix^{-1/2} = \bigg(\iMatrix - \frac{\rMatrix^{1/2} \circVec_{e} \circVec_{e}^T \rMatrix^{1/2}}{\| \circVec_{e} \|^2_\rMatrix} \bigg)\enspace.
$$
This shows that $\Proj_e$ is an orthogonal projection. Equivalently, before the scaling, the operator $\Pi_e$ is an orthogonal projection with respect to the inner product defined by $\rMatrix.$

The main theorem of this section involves showing a relation between the operator
$\left(\prod_{j = 1}^\nIter \Proj_{e_j} \right),$ which describes the cycle updates performed by $\simpleSolver$, and the projection operator $\Pi_G \in \R^{E \times E},$ defined as:
$$
\Pi_G \defeq \rMatrix^{-1/2} \incMatrix \pseudo{\lap} \incMatrix^T \rMatrix^{-1/2}\enspace.
$$
The projection $\Pi_G$ takes any flow $\rMatrix^{1/2} \flow$ and projects it orthogonally onto the subspace of KPL-respecting flows, i.e. electrical flows.
The operator $\Pi_G$ plays an important role in the spectral study of graphs and is the main object of study in spectral sparsification~\cite{spielman2011graph}.

\begin{theorem}\label{thm:projapprox}
Let $e_1, \ldots, e_\nIter$ be random off-tree edges picked independently from the $\{p_e\}$
probability distribution.
Let $\nIter = \lceil \treeCondition \log \frac{n}{\epsilon p} \rceil.$
Then, with probability at least $1-p$,
$$
\norm{\left(\prod_{j = 1}^\nIter \Proj_{e_j} \right) - \Pi_G}^2_F \leq \epsilon \enspace.
$$
\end{theorem}

We defer the proof to the end of this section and first show how two useful bounds involving the flow and voltage operators follow directly from \theoremref{thm:projapprox}.

\begin{theorem}\label{thm:lapapprox}
Let $e_1, \ldots, e_\nIter$ be random off-tree edges picked independently from the $\{p_e\}$
probability distribution.
Let $\nIter = \lceil \treeCondition \log \frac{n}{\epsilon p} \rceil.$
Then, with probability at least $1-p$,
\begin{equation} \label{eq:flow-operator-approx}
(\mvar{F}_\nIter - \rMatrix^{-1} \incMatrix \pseudo{\lap})^T \rMatrix (\mvar{F}_\nIter - \rMatrix^{-1} \incMatrix \pseudo{\lap})
\preceq \epsilon \pseudo{\lap_T} \preceq \epsilon \stretchTotal{T} \pseudo{\lap}
\end{equation}
and
\[
(\mvar{V}_\nIter - \pseudo{\lap})^T \lap (\mvar{V}_\nIter - \pseudo{\lap}) \preceq \epsilon (\stretchTotal{T})^2 \pseudo{\lap}\enspace.
\]
\end{theorem}

In particular, the first equation implies that, for $\delta > 0$ with probability at least $1-p$, a random sequence of $\nIter' = \lceil \tau \log \frac{n \cdot \stretchTotal{T}}{\delta p} \rceil$ cycle updates yield a flow operator $\mvar{F}_{\nIter'}$ which is a multiplicative $\delta$-approximation to the electrical-flow operator $\rMatrix^{-1} \incMatrix \pseudo{\lap}.$ In practice, this means that, if we choose $\nIter'$ cycles at random according to $\{p_e\},$ the resulting sequence of updates will yield a flow whose energy is at most a factor of $(1+\delta)$ away from optimum, {\it for all} possible initial demand vectors $\boundary.$

Similarly, the second equation implies that $\nIter' = \lceil \tau \log \frac{n \cdot \stretchTotal{T}^2}{\delta p} \rceil$ randomly selected cycle updates yields with probability at least $1 - p$ a voltage operator that when applied to any demand $\boundary$ results in a $\sqrt{\delta}$-approximate solution to $\lap \vvar{\volt} = \boundary$.

\begin{proof}[Proof of \theoremref{thm:lapapprox}]
\theoremref{thm:projapprox} implies:
$$
\left(\left(\prod_{j = 1}^\nIter \Proj_{e_j} \right) - \Pi_G\right)^T \left(\left(\prod_{j = 1}^\nIter \Proj_{e_j} \right) - \Pi_G\right) \preceq \epsilon \iMatrix \enspace.
$$
It suffices to left- and right-multiply this inequality by $ \rMatrix^{-1/2} \incMatrix_T^T \pseudo{\lap_T} $ to obtain:
$$
(\mvar{F}_\nIter - \rMatrix^{-1} \incMatrix \pseudo{\lap})^T \rMatrix (\mvar{F}_\nIter - \rMatrix^{-1} \incMatrix \pseudo{\lap}) \preceq \epsilon \pseudo{\lap_T}
$$
as required. To obtain the inequality involving $\pseudo{\lap},$ we can directly apply the fact that $\lap \preceq \stretchTotal{T} \lap_T$ (cf. \cite[Lemma 9.2]{ST08c}).

We now derive the bound for the voltage operator. This requires the following chain of inequalities:
$$
\rMatrix \succeq \rMatrix^{1/2} (\rMatrix^{-1/2} \incMatrix_T \pseudo{\lap_T} \incMatrix_T^T \rMatrix^{-1/2}) \rMatrix^{1/2} = \incMatrix_T \pseudo{\lap_T} \incMatrix_T^T \succeq \frac{1}{\stretchTotal{T}} \incMatrix_T \pseudo{\lap_T} \lap \pseudo{\lap_T}  \incMatrix_T^T \enspace.
$$
The first inequality follows as $\rMatrix^{-1/2} \incMatrix_T \pseudo{\lap_T} \incMatrix_T^T \rMatrix^{-1/2}$ is an orthogonal projection.
The second one is again a consequence of the fact that  $\lap \preceq \stretchTotal{T} \lap_T$.
It now suffices to replace $\rMatrix$ by $\frac{1}{\stretchTotal{T}} \incMatrix_T \pseudo{\lap_T} \lap \pseudo{\lap_T}  \incMatrix_T^T$ in \equationref{eq:flow-operator-approx} to obtain the required bound.
\end{proof}

The rest of the section is dedicated to proving \theoremref{thm:projapprox}.
This could be achieved by slightly modifying the proof of \sectionref{sec:convergence}. However, in the following, we prefer to briefly give an operator-based proof of the theorem. We believe that this complements the view of \sectionref{sec:convergence} and helps to convey its geometric flavor.

\subsubsection{Proof of \texorpdfstring{\theoremref{thm:projapprox}}{Projection Operator Approximation}}

The spanning tree $\tree$ allows us to construct a projector $\Pi_\tree$ that can be used to relate $\Pi_G$ and $\left(\prod_{j = 1}^\nIter \Proj_{e_j} \right).$
The projector $\Pi_\tree$ is defined as
$$
\Pi_\tree \defeq \rMatrix^{-1/2} \incMatrix \pseudo{\lap_T} \incMatrix_{\tree}^T \rMatrix^{-1/2}\enspace.
$$
Here, it is important to notice that $\Pi_T$ is not symmetric, i.e. it is not an orthogonal projector, but an oblique one. The range of $\Pi_T$ is the range of $\rMatrix^{-1/2} \incMatrix^T$ i.e.  the space of KPL-respecting flows.
The following lemma relates $\Pi_T$ to the tree cycle vectors $\{\circVec_e\}_{e \in \offtreeEdgeSet}.$
\begin{lemma}\label{lem:tree-projection}
Let $\vec{1}_e \in \R^E$ be the standard unit vector associated to $e \in E.$
Then:
$$
(\iMatrix - \Pi_T)^T \rMatrix^{1/2}  \vec{1}_e  =
\begin{cases}
 \rMatrix^{1/2} \circVec_e , &\enspace e \in \offtreeEdgeSet\enspace,\\
0, &\enspace e \in T \enspace.
\end{cases}
$$
\end{lemma}
\begin{proof}
For $e \in T,$ we have:
$$
(\iMatrix - \Pi_T)^T \rMatrix^{1/2}  \vec{1}_e  =
\rMatrix^{1/2}  \vec{1}_e - \rMatrix^{-1/2} \incMatrix^T \pseudo{\lap_T} \incMatrix_T \vec{1}_e = \rMatrix^{1/2}  (\vec{1}_e - \vec{1}_e) = 0\enspace.
$$
For $e \in \offtreeEdgeSet:$
$$
(\iMatrix - \Pi_T)^T \rMatrix^{1/2}  \vec{1}_e  =
\rMatrix^{1/2}  \vec{1}_e - \rMatrix^{-1/2} \incMatrix^T \pseudo{\lap_T} \incMatrix_T \vec{1}_e = \rMatrix^{1/2}  (\vec{1}_e - \vec{p}_e) = \rMatrix^{1/2}  \circVec_e\enspace.
$$
\end{proof}
In words, the rows corresponding to off-tree edges in the matrix representation of $\rMatrix^{1/2} (\iMatrix - \Pi_T) \rMatrix^{-1/2}$ in the standard basis are exactly the tree cycle vectors $\{\circVec_e\}_{e \in \offtreeEdgeSet}.$ The remaining rows are $0.$

In the next lemma, we relate $\Pi_T$ to the cycle projections $\{\Proj_e\}_{e \in \offtreeEdgeSet}$ used in $\simpleSolver$.
\begin{lemma} \label{lem:projexpectation}
Let $e \in \offtreeEdgeSet$ be sampled according to probability $\{p_e\},$ as in $\simpleSolver.$ Then:
\[
\E[\Proj_e] = \iMatrix - \frac{1}{\tau} (\iMatrix - \Pi_T)^T(\iMatrix - \Pi_T)\enspace.
\]
\end{lemma}
\begin{proof}
This follows from the definition of $\{p_e\}$ and \lemmaref{lem:tree-projection}:
$$
\E[\Proj_e] = \iMatrix - \frac{1}{\tau}  \sum_{e \in \offtreeEdgeSet} \frac{ \rMatrix^{1/2} \circVec_e \circVec^T_e \rMatrix^{1/2} }{r_e} =  \iMatrix - \frac{1}{\tau} (\iMatrix - \Pi_T)^T(\iMatrix - \Pi_T)\enspace.
$$
\end{proof}

Next, we show how to relate $(\iMatrix - \Pi_T)^T(\iMatrix - \Pi_T)$ to $\iMatrix - \Pi_G.$
We will show that, for any flow $\rMatrix^{1/2} \flow,$ the distance  $\big\|\rMatrix^{1/2} \flow - \Pi_G  \rMatrix^{1/2}  \flow\big\|_2$ to its $\Pi_G$-projection is smaller than $\big\|\rMatrix^{1/2} \flow - \Pi_T \rMatrix^{1/2}  \flow \big\|_2,$ the distance to its $\Pi_T$-projection.
It is easy to see that this should be the case, because $\Pi_G$ and $\Pi_T$ project onto the same subspace and $\Pi_G$ is an orthogonal projection. We prove it formally here. This is analogous to \lemmaref{lemma:geocondition}.
\begin{lemma}
\label{lem:tree-projection-2}
$$
(\iMatrix - \Pi_T)^T(\iMatrix - \Pi_T) \succeq (\iMatrix - \Pi_G)^T (\iMatrix - \Pi_G) = \iMatrix - \Pi_G\enspace.
$$
\end{lemma}
\begin{proof}
Algebraically, we can check that
\begin{align*}
(\Pi_G - \Pi_T)^T (\Pi_G - \Pi_T) + (\iMatrix - \Pi_G) =
\Pi_G - \Pi_T^T \Pi_G - \Pi_G \Pi_T + \Pi_T^T \Pi_T + \iMatrix - \Pi_G =\\
 \iMatrix - \Pi_T^T - \Pi_T + \Pi_T^T \Pi_T = (\iMatrix - \Pi_T)^T(\iMatrix - \Pi_T)\enspace.
\end{align*}
In the second equality, we used the fact that $\Pi_G \Pi_T = \Pi_T$ as $\Pi_G$ and $\Pi_T$ project onto the same subspace.
As $(\Pi_G - \Pi_T)^T (\Pi_G - \Pi_T) \succeq 0,$ the lemma follows.
\end{proof}

We are now ready to prove \theoremref{thm:projapprox}.
\begin{proof}[Proof of \theoremref{thm:projapprox}]
For any iteration $t,$ $0 \leq t \leq \nIter,$ let
$$
\Phi_t \defeq \norm{\left(\prod_{j = 1}^t \Proj_{e_j} \right) - \Pi_G}^2_F\enspace.
$$
We will show that $\Phi_\nIter \leq \epsilon$ with probability at least $1-p,$ by bounding the expected reduction in $\Phi_t$ at every iteration.
Notice that for any $e \in \offtreeEdgeSet,$ we have $\Pi_G \Proj_e = \Proj_e \Pi_G = \Pi_G,$ as $\Pi_G$ and $\Proj_e$ are both orthogonal projectors and the range of $\Pi_G$ is contained into that of $\Proj_e.$
Hence, we can simplify the expression for $\Phi_t$ as
$$
\Phi_t = \Tr\left(\left(\prod_{j = 1}^t \Proj_{e_j} \right)^T \left(\prod_{j = 1}^t \Proj_{e_j} \right) - \Pi_G\right)\enspace.
$$

Now, we are ready to bound the change in $\Phi_t$ in one iteration. For any $t,$ $t \leq 0 \leq \nIter -1 ,$ assume that $e_1, \cdots, e_t$ have been fixed and notice:
\begin{align*}
\E[\Phi_{t+1}] = \E\left[\Tr\left(\left(\prod_{j = 1}^{t+1} \Proj_{e_j} \right)^T \left(\prod_{j = 1}^{t+1} \Proj_{e_j} \right) - \Pi_G\right)\right] =
 \Tr\left(\left(\prod_{j = 1}^t \Proj_{e_j} \right)^T \E[\Proj_{e_{t+1}}] \left(\prod_{j = 1}^t \Proj_{e_j} \right) - \Pi_G\right) \enspace.
\end{align*}
By \lemmaref{lem:projexpectation}, we obtain:
$$
\E[\Phi_{t+1}] = \Phi_t - \frac{1}{\tau}
\Tr\left(\left(\prod_{j = 1}^t \Proj_{e_j} \right)^T (\iMatrix - \Pi_T)^T(\iMatrix - \Pi_T) \left(\prod_{j = 1}^t \Proj_{e_j} \right)\right)\enspace.
$$
We now apply \lemmaref{lem:tree-projection-2}:
$$
\E[\Phi_{t+1}] \leq \Phi_t - \frac{1}{\tau}
\Tr\left(\left(\prod_{j = 1}^t \Proj_{e_j} \right)^T (\iMatrix - \Pi_G)
 \left(\prod_{j = 1}^t \Proj_{e_j} \right)\right)\enspace.
$$
This yields that
$
\E[\Phi_{t+1}] \leq \left(1 -\frac{1}{\tau}\right) \Phi_t
$
and by taking expectation over all the choices of sampled edges, we have
$
\E[\Phi_\nIter] \leq \left(1 - \frac{1}{\tau}\right)^\nIter \Phi_0.
$
Now, by definition, $\Phi_0 = \norm{\Pi_G}^2_F$ and because $\Pi_G$ is a projector onto an $(n-1)$ dimensional space, we have $\Phi_0 = n - 1.$
Hence, for our choice of $\nIter = \lceil\tau \log\frac{n}{\epsilon p}\rceil,$
we have:
$$
\E[\Phi_\nIter] \leq \left(1 - \frac{1}{\tau}\right)^\nIter (n-1) \leq \epsilon p\enspace.
$$
By Markov's Inequality, with probability at least $1-p,$ it must be the case that $\Phi_\nIter \leq \epsilon.$
\end{proof}

\section{Solving without Low Stretch Spanning Trees}
\label{sec:no-spanning-tree}

In the algorithms presented in this paper the spanning tree, $\tree$, provides two essential functions. (1) The $\tree$ induces a well conditioned basis for cycle space and (2) the spanning tree allows for a compact representation of cycles that allow efficient querying and updating. It is natural to ask, can a combinatorial object other than a low-stretch spanning subgraph be used to achieve these goals? To the best of our knowledge all previous nearly-linear-time SDD-system solvers use such low-stretch spanning trees; can a nearly-linear running time be achieved without them?

All these questions can be answered in the affirmative and in fact a particular \emph{decomposition tree} (low-stretch tree approximation of $G$ that are not necessarily subgraphs)\cite{bartal96, Bartal98, FRT03} construction of Bartal can used precisely for this purpose. Several steps are needed. First, we need to show how decomposition trees in general can be used to satisfy property (1). Next we need to show how we can replicate edges to create an equivalent problem without hurting the property (1) too much. Finally, we need to show how this replication can be done to achieve just enough path disjointness that we can apply the data structure from \sectionref{sec:datastructure} and achieve property (2). Ultimately, this can all be accomplished while only losing several log factors in the run time (as compared to the most efficient algorithm presented in this paper). \emph{(We plan on including the details in a later version of this paper.)}

\section{Acknowledgments}
We thank Daniel Spielman for many helpful conversations.  This work was partially supported by NSF awards 0843915 and 1111109, a Sloan Research Fellowship, and a NSF Graduate Research Fellowship (grant no. 1122374).

\bibliographystyle{alpha}
\bibliography{eflow_full}

\appendix
\clearpage
\section{Reduction from SDD Systems to Laplacian Systems}
\label{appendix:reductions}

Reductions from SDD system solving to Laplacian system solving can be traced back to \cite{GrembanPHD} where the problem of solving a linear system in a \emph{symmetric diagonally dominant} (SDD) matrix of dimension $n$ was reduced to that of solving a Laplacian matrix of dimension $2n+1$. This reduction contains two steps. First, a SDD matrix is reduced to a SDDM matrix of dimension $2n$, where a SDDM matrix is an SDD matrix with non-positive off-diagonals.\footnote{An SDDM matrix is not necessarily Laplacian since its rows (or columns) may not sum up to zeros. For historical reasons, Gremban~\cite{GrembanPHD} uses the notion ``generalized Laplacian matrix'' for SDD matrix, and the notion ``Laplacian matrix'' for SDDM matrix.}. Second, this SDDM matrix is augmented with an extra dimension yielding a Laplacian matrix of dimension $2n+1$. Following the interpretation of a Laplacian system as an electrical circuit system (see~\sectionref{sec:pre-duality}) this second step can also be viewed as adding an extra ground vertex with fixed zero voltage, and connecting it to all other vertices whose corresponding rows do not sum up to zero.

The first step of the above reduction is often referred to as \emph{Gremban's reduction}, and has been used in previous nearly-linear-time SDD system solvers (see for instance~\cite{ST08c}). Although this reduction is stated for exact solutions, it also carries approximate solutions from one to the other~\cite{ST08c}.

Below we present a variant of Gremban's reduction, that in a single step reduces from a linear system of a SDD matrix directly to that of a Laplacian matrix and we prove that this reduction preserves approximate solutions. Since for Laplacians of disconnected graphs each component can be solved independently we see that it suffices to study Laplacian matrices for connected graphs in the main body of our paper.

\subsection{Our Reduction}

Given an arbitrary SDD matrix $\mvar{A}$, we can always decompose it into $\mvar{A}=\mvar{D}_1+\mvar{A}_p+\mvar{A}_n+\mvar{D}_2$, where
\begin{itemize}[nolistsep]
\item $\mvar{A}_p$ is the matrix containing all positive off-diagonal entries of $\mvar{A}$,
\item $\mvar{A}_n$ is the matrix containing all negative off-diagonal entries of $\mvar{A}$,
\item $\mvar{D}_1$ is the diagonal matrix where $\mvar{D}_1(i,i) \defeq \sum_{j=1}^n |\mvar{A}(i,j)|$, and
\item $\mvar{D}_2 = \mvar{A} - \mvar{A}_p - \mvar{A}_n - \mvar{D}_1$ is the excess diagonal matrix.
\end{itemize}
Given that $\mvar{A}$ is SDD it is a simple exercise to check that $\mvar{D}_1 + \mvar{A}_n - \mvar{A}_p$ is Laplacian and the entries of $\mvar{D}_2$ are non-negative. Now, consider the following linear system
\begin{equation*}
\boxed{
\tilde{\mvar{A}}
\begin{pmatrix} x_1 \\ x_2 \\\end{pmatrix}
= \begin{pmatrix} b \\ -b \\\end{pmatrix}
}
\quad \text{ where }
\tilde{\mvar{A}} = \begin{pmatrix} \mvar{D}_1 + \mvar{D}_2/2 + \mvar{A}_n & -\mvar{D}_2/2 -\mvar{A}_p \\ -\mvar{D}_2/2 -\mvar{A}_p & \mvar{D}_1 + \mvar{D}_2/2 + \mvar{A}_n \\\end{pmatrix}
\enspace.
\end{equation*}
It is not hard to verify that $\tilde{\mvar{A}}$ is Laplacian and that given an exact solution $(x_1, x_2)$ to the above linear system (and notice that there may be many such solutions) we have
\begin{equation*}
\left.
  \begin{array}{ll}
    (\mvar{D}_1 + \mvar{D}_2/2 + \mvar{A}_n)x_1 + (-\mvar{D}_2/2-\mvar{A}_p) x_2 = b \\
    (-\mvar{D}_2/2 -\mvar{A}_p)x_1 + (\mvar{D}_1 + \mvar{D}_2/2 + \mvar{A}_n) x_2 = -b
  \end{array}
\right\}
\text{ implies }
\mvar{A} \left(\frac{x_1 - x_2}{2} \right)=b \enspace.
\end{equation*}
So $x=\frac{x_1 - x_2}{2}$ is also an exact solution to the original linear system $\mvar{A} x = b$.

Now let us verify that this reduction also holds for approximate solutions. Let $(\hat{x}_1,\hat{x}_2)$ denote an arbitrary $\epsilon$-approximate solution to the new Laplacian system, and let $\hat{e}_1 \defeq \hat{x}_1 - x_1$ and $\hat{e}_2 \defeq \hat{x}_2-x_2$. By definition we have:
\begin{equation*}
\left \| \begin{pmatrix} \hat{e}_1 \\ \hat{e}_2 \\\end{pmatrix} \right\|_{\tilde{\mvar{A}}}
=
\left \| \begin{pmatrix} \hat{x}_1 \\ \hat{x}_2 \\\end{pmatrix} -  \begin{pmatrix} x_1 \\ x_2 \\\end{pmatrix}  \right\|_{\tilde{\mvar{A}}}
\leq
\epsilon \left \|  \begin{pmatrix} x_1 \\ x_2 \\\end{pmatrix} \right\|_{\tilde{\mvar{A}}}
\end{equation*}
and by expanding it out we have:

\begin{align}
&&\begin{pmatrix} \hat{e}_1 \\ \hat{e}_2 \\\end{pmatrix}^T \tilde{\mvar{A}} \begin{pmatrix} \hat{e}_1 \\ \hat{e}_2 \\\end{pmatrix}
 &\leq \epsilon^2
\begin{pmatrix} x_1 \\ x_2 \\\end{pmatrix}^T \tilde{\mvar{A}} \begin{pmatrix} x_1 \\ x_2 \\\end{pmatrix} \nonumber\\
\Leftrightarrow
&&\begin{pmatrix} \hat{e}_1 \\ \hat{e}_2 \\\end{pmatrix}^T
\begin{pmatrix} \mvar{D}_1 + \mvar{D}_2/2 + \mvar{A}_n & -\mvar{D}_2/2 -\mvar{A}_p \\ -\mvar{D}_2/2 -\mvar{A}_p & \mvar{D}_1 + \mvar{D}_2/2 + \mvar{A}_n \\\end{pmatrix}
\begin{pmatrix} \hat{e}_1 \\ \hat{e}_2 \\\end{pmatrix}
 &\leq \epsilon^2
\begin{pmatrix} x_1 \\ x_2 \\\end{pmatrix}^T \begin{pmatrix} b \\ -b \\\end{pmatrix}
= 2\epsilon^2 x^T b \label{eqn:reduction-1}
\end{align}

Now, because $\mvar{D}_1 + \mvar{A}_n - \mvar{A}_p$ is Lapacian (and therefore positive semidefinite) we have
\begin{align*}
&&\begin{pmatrix}  \mvar{D}_1 + \mvar{A}_n - \mvar{A}_p & \mvar{D}_1 + \mvar{A}_n - \mvar{A}_p \\ \mvar{D}_1 + \mvar{A}_n - \mvar{A}_p & \mvar{D}_1 + \mvar{A}_n - \mvar{A}_p \end{pmatrix}
&\succeq 0 \\
\Leftrightarrow&&
\begin{pmatrix}  \mvar{D}_1 + \mvar{A}_n - \mvar{A}_p & \mvar{D}_1 + \mvar{A}_n - \mvar{A}_p \\ \mvar{D}_1 + \mvar{A}_n - \mvar{A}_p & \mvar{D}_1 + \mvar{A}_n - \mvar{A}_p \end{pmatrix}
+ \begin{pmatrix} \mvar{A} & -\mvar{A} \\ -\mvar{A} & \mvar{A} \end{pmatrix}
&\succeq \begin{pmatrix} \mvar{A} & -\mvar{A} \\ -\mvar{A} & \mvar{A} \end{pmatrix} \\
\Leftrightarrow&&
2\begin{pmatrix}  \mvar{D}_1 + \mvar{D}_2/2 + \mvar{A}_n & -\mvar{D}_2/2 - \mvar{A}_p \\ -\mvar{D}_2/2 - \mvar{A}_p & \mvar{D}_1 + \mvar{D}_2/2 + \mvar{A}_n \end{pmatrix}
&\succeq \begin{pmatrix} \mvar{A} & -\mvar{A} \\ -\mvar{A} & \mvar{A} \end{pmatrix}
\end{align*}
and as a consequence we deduce from \equationref{eqn:reduction-1} that (by defining $\hat{x}=\frac{\hat{x}_1-\hat{x}_2}{2}$)
\begin{align*}
\frac{1}{2}\begin{pmatrix} \hat{e}_1 \\ \hat{e}_2 \\\end{pmatrix}^T
\begin{pmatrix} \mvar{A} & -\mvar{A} \\ -\mvar{A} & \mvar{A} \end{pmatrix}
\begin{pmatrix} \hat{e}_1 \\ \hat{e}_2 \\\end{pmatrix}
\leq 2\epsilon^2 x^T b
\quad&\Leftrightarrow \quad
2 \left(\frac{\hat{e}_1 - \hat{e}_2}{2}\right)^T \mvar{A} \left(\frac{\hat{e}_1 - \hat{e}_2}{2}\right) \leq 2\epsilon^2 x^T b \\
\quad&\Leftrightarrow \quad
2 (\hat{x}-x)^T \mvar{A} (\hat{x}-x) \leq 2\epsilon^2 x^T b \\
\quad&\Leftrightarrow \quad
\|\hat{x}-x \|_{\mvar{A}} \leq \epsilon \|x\|_{\mvar{A}} \enspace.
\end{align*}
Therefore, we have shown that $\hat{x}=\frac{\hat{x}_1-\hat{x}_2}{2}$ is also an $\epsilon$ approximate solution to $\mvar{A}x = b$.

\end{document}